\documentclass[aps,prx,nopacs,nokeys,superscriptaddress,11pt,oneside,a4paper]{revtex4-2}

\usepackage[utf8]{inputenc}

\usepackage{xcolor}
\usepackage{amsmath,amsthm}
\usepackage{amsfonts}
\usepackage{amssymb}
\usepackage{makeidx}
\usepackage{graphicx}
\usepackage{tikz}
\usepackage{tikz-cd}
\usepackage{mathrsfs}

\usepackage[caption=false]{subfig}

\usepackage{color}
\usepackage[]{hyperref}
\usepackage{url}

\DeclareMathOperator{\tr}{\mathrm{tr}}
\DeclareMathOperator{\proj}{\mathrm{proj}}
 
\theoremstyle{plain}
\newtheorem{theorem}{Theorem}

\newtheorem{proposition}{Proposition}
\newtheorem{corollary}{Corollary}

\theoremstyle{definition}
\newtheorem{definition}{Definition}
\newtheorem{postulate}{Postulate}

\newtheoremstyle{boldremark}
  {3pt}    
  {3pt}    
  {\normalfont} 
  {}       
  {\bfseries} 
  {.}      
  {0.5em}  
  {}       

\theoremstyle{boldremark}
\newtheorem{remark}{Remark}
\newtheorem{example}{Example}

\newcommand{\sech}{\mathrm{sech}}

\begin{document}

\title{Geometric foundations of thermodynamics in the quantum regime}
\author{\'Alvaro Tejero}\email[]{atejero@onsager.ugr.es}\affiliation{Quantum Thermodynamics and Computation Group, Departamento de Electromagnetismo y Física de la Materia, Universidad de Granada, 18071 Granada, Spain\looseness=-1}
\affiliation{Instituto Carlos I de Física Teórica y Computacional, Universidad de Granada, 18071 Granada, Spain\looseness=-1}
\author{Martín de la Rosa}\email[]{f42rodim@uco.es}\affiliation{Departamento de Matemáticas, Edificio Albert Einstein, Universidad de Córdoba, Campus de Rabanales 14071, Córdoba, Spain\looseness=-1}

\begin{abstract}
In this work, we present a geometrical formulation of quantum thermodynamics based on contact geometry and principal fiber bundles. The quantum thermodynamic state space is modeled as a contact manifold, with equilibrium Gibbs states forming a Legendrian submanifold that encodes the fundamental thermodynamic relations. A principal fiber bundle over the manifold of density operators distinguishes the quantum state structure from thermodynamic labels: its fibers represent non-equilibrium configurations, and their unique intersections with the equilibrium submanifold enforce thermodynamic consistency. Quasistatic processes correspond to minimizing geodesics under the Bures–Wasserstein metric, leading to minimal dissipation, while the divergence of geodesic length toward rank-deficient states geometrically derives the unattainability aspect of the third law. Non-equilibrium extensions, formulated through pseudo-Riemannian metrics and connections on the principal bundle, introduce curvature-induced holonomy that quantifies a geometric source of irreversibility in cyclic processes. In this framework, the thermodynamic laws in the quantum regime emerge naturally as geometric consequences.
\end{abstract}

\maketitle

\section{Introduction}
Physics is geometry. Differential geometry~\cite{doCarmo1992Riemannian,ONeill1983SemiRiemannian,caja12_variedades,rudolph2012_1,rudolph2012_2} structures are fundamental to physical theories, from the symplectic geometry of classical mechanics~\cite{arnold1989mathematical} to the curved spacetime of general relativity~\cite{wald1984general} and the fiber bundle formulation of gauge theories~\cite{Bleecker1981Gauge,naber2011_1,naber2011_2,hamilton2017mathematical}. 
In thermodynamics~\cite{callen85_thermodynamics}, contact geometry provides a natural and rigorous framework for equilibrium states and processes, as originally conceived by Gibbs and systematically developed in modern terms~\cite{bravetti15_contact,bravetti2017contact}.
Here, the thermodynamic state space is a $(2n+1)$-dimensional contact manifold $(M, \eta)$, with the Gibbs 1-form $\eta$ encoding the first law, and equilibrium states forming a Legendrian submanifold $E$ where $\eta|_E = 0$. 
Extending this geometric structure to the quantum regime~\cite{binder2018thermodynamics, strasberg2022quantum} presents significant conceptual and mathematical challenges.

Geometric methods have been increasingly applied to quantum thermodynamics, drawing from both differential geometry and information geometry to analyze dissipation, optimal processes, and fluctuations. In information geometry, the space of probability distributions (classical) or density operators (quantum) is endowed with a Riemannian metric, often the Fisher information metric or its quantum analogs like the quantum Fisher information \cite{braunstein_94, amari2000methods}. This framework has been used to link geometric quantities to thermodynamic concepts: for instance, the Fisher metric relates to thermalization speed and entropy production rates in stochastic thermodynamics \cite{Crooks2007, ibanez2024, tejero2025asymmetries, melo_2025, tejero2025entropy, Bettmann2025Information}, and the quantum Fisher information bounds dissipation in finite-time quantum processes \cite{deffner_13}.
Other approaches emphasize Riemannian or pseudo-Riemannian metrics on the space of quantum states for optimization. For example, thermodynamic length---originally from classical thermodynamics \cite{weinhold_75, Ruppenier_79, salamon_83}---has been generalized to quantum systems \cite{scandi2019thermodynamic, abiuso_20,Bettmann2025Information, anza2022geometric}.

Despite these advances, a complete and unified geometric formulation of the full theory is still lacking. In the absence of a geometric framework analogous to that of the classical case, existing analyses remain fragmented, and general principles are difficult to establish rigorously.
The aforementioned works have investigated quantum thermodynamics using \emph{differential geometry}, or tools drawn from it, such as representing processes as curves or defining thermodynamic lengths from metrics on the state space. However, these approaches do not constitute a proper geometric formulation of the theory; rather, they provide geometric perspectives on specific classes of processes.
The geometric formulation proposed here builds on some of these ideas by reformulating the quantum state space within the language of contact geometry, extending classical contact thermodynamics, and incorporating the Bures–Wasserstein metric to describe quasistatic processes. These elements naturally connect with earlier uses of Riemannian metrics \cite{Ruppenier_79}.

Principal fiber bundles, central to gauge theories and topological phases in condensed matter~\cite{nakahara2018geometry,Bleecker1981Gauge}, offer a powerful global structure for addressing these gaps. When applied to quantum thermodynamics, they furnish a principal bundle over the space of density operators, with fibers representing the thermodynamic coordinates associated with a fixed quantum state, thereby separating state evolution from thermodynamic variation and incorporating gauge-like symmetries.

In this paper, we synthesize these elements into a geometrical framework extended to the quantum regime. The quantum thermodynamic state space is indeed a contact manifold $\mathcal{M}$, where thermal equilibrium states form the Legendrian submanifold $\mathcal{E} \subset\mathcal{M}$, satisfying the quantum first law.
We introduce a principal fiber bundle over the space of density operators, mapping thermodynamic coordinates to states. Fibers enable analysis of relaxation paths, with equilibrium as unique intersections with the equilibrium submanifold.
The Bures-Wasserstein metric on such submanifold~\cite{bengtsson2006geometry,oostrum22_bures} defines geodesics as quasistatic processes minimizing dissipation. Boundary effects in finite dimensions lead to geodesic incompleteness, recovering the third law geometrically.
Non-equilibrium phenomena are addressed naturally in the theory, via pseudo-Riemannian extensions of the metric, compatible with $\ker \eta$, and the structure of the fiber bundle. 
The principal connection decomposes the tangent bundle, with curvature inducing holonomy in cyclic processes, a geometric source of irreversibility akin to gauge holonomies.

With all these considerations, we construct a principal fiber bundle over the manifold of thermal states, separating quantum state structure from thermodynamic labels, enabling a unified treatment of equilibrium (Legendrian submanifold) and non-equilibrium (fiber paths). Then, the framework allows the understanding of the zeroth, first, second, and third laws as consequences of bundle properties, injectivity, contact non-integrability, and geodesic incompleteness, respectively. Finally, the curvature-induced holonomy acts as a gauge-like source of irreversibility in cyclic processes.

This geometrization not only generalizes classical geometric thermodynamics but also provides mathematical rigor for quantum thermodynamics, going beyond process-specific analyses by offering a geometric theory analogous to gauge theories. 
While the core framework developed here provides a geometric formulation for finite-dimensional quantum systems with a finite set of observables---where mathematical structures such as the principal bundle and Bures–Wasserstein metric are fully rigorous---we also discuss preliminary extensions to infinite-dimensional cases.

\section{Motivation: classical thermodynamics}
Physics can be regarded as a collection of models that attempt to describe, always with some degree of idealization, specific aspects of the observable universe.
Equilibrium thermodynamics exemplifies this as a robust physical theory centered on one of the most abstract idealizations: every system is assumed to be perpetually in thermodynamic equilibrium, implying that such a state would remain unchanged indefinitely. 
Paradoxically, this framework addresses phenomena like gas expansions, an out-of-equilibrium process.
This is reconciled by positing that the evolution of the system proceeds exclusively through successive equilibrium states, necessitating an infinite duration for completion---a process termed \emph{quasistatic}.
Consequently, thermodynamic parameters vary smoothly, allowing well-defined derivatives with respect to an evolution parameter. These variations occur so gradually that they can be treated as infinitesimal, explaining why thermodynamic equations are typically expressed in differential forms.
For instance, the first law is commonly written as $ dU = T dS + \sum_i X_i  dY_i + \sum_j \mu_j dN_j. $
The terms $dS$, $dY_i$, $dN_j$, for all $i,j$, can be interpreted as mere infinitesimals leading to finite integrals or, more insightfully, as differential 1-forms.

The latter perspective lends natural clarity to many thermodynamic concepts and results.
Standard curricula on this topic usually distinguish state functions like entropy and internal energy---which depend solely on the system state when measured---from path-dependent quantities like heat or work. 
Thus, changes in internal energy or entropy are independent of the trajectory in thermodynamic parameter space, depending only on the endpoints, making them exact differentials.
In contrast, integrals of heat or work generally vary with the path. A non-exact differential form may become exact when multiplied by an integrating factor, provided the form is closed. Assuming simply connected manifolds, closed and exact forms coincide. The prototypical example is heat: by multiplying the infinitesimal amount $\delta Q$ by $1/T$ one obtains the exact differential $dS$. 

Carathéodory's theorem, often disregarded in classical treatments, emerges straightforwardly geometrically: in any neighborhood of a thermodynamic state, certain states are inaccessible via quasistatic adiabatic processes. Geometrically, a thermodynamic process is a curve $\gamma: I \subseteq \mathbb{R} \to M$, where $M$ is the thermodynamic state space endowed with a differentiable manifold structure, as later explained. If $\delta Q \in \Lambda^1(M)$, the curve $\gamma$ is adiabatic when $\delta Q(\gamma') = 0$ everywhere. Since $dS = \delta Q / T$, it follows that $\delta Q = T dS$, so 
\begin{equation}
    \delta Q(\gamma') = T dS(\gamma') = T \gamma'(S) = T \frac{d}{dt} (S \circ \gamma) = 0.
\end{equation}
Hence, quasistatic adiabatic processes cannot connect states of differing entropies.

The geometrization of classical thermodynamics provides a powerful framework for analyzing thermodynamic systems using tools from differential geometry. 
This approach represents thermodynamic systems as manifolds endowed with specific geometric structures, enabling a rigorous description of equilibrium states and thermodynamic processes.
Here, we introduce the thermodynamic state space and its contact structure, establishing the foundation for extending these concepts to quantum thermodynamics in subsequent sections.

The thermodynamic state space is formalized as a $(2n+1)$-dimensional manifold $M$, where $n$ is the number of independent extensive variables characterizing the system. This manifold is equipped with a contact structure defined by a 1-form $\eta$, satisfying the non-integrability condition $\eta \wedge (d\eta)^n \neq 0$, where $\wedge$ denotes the wedge (also exterior or Grassmann) product and $(d\eta)^n$ is the $n$-fold exterior product of $d\eta$. The contact structure ensures that $\eta$ defines a maximally non-integrable distribution that captures the constraints of thermodynamic processes.
In physical terms, the thermodynamic state space is endowed with the Gibbs 1-form, typically expressed for a simple thermodynamic system as
\begin{equation}\label{eqn:gibbs}
    \eta = dU - T dS - \sum_{i=1}^{m} X_i dY_i - \sum_{j=1}^k \mu_j  dN_j,
\end{equation}
where $U$ is the internal energy, $S$ is the entropy, $Y_i$ are the extensive variables, $X_i$ are their conjugate intensive variables, $N_j$ are the particle numbers for $k$ species, and $T$ and $\mu_j$ are the temperature and chemical potentials, respectively, for $i=1,\ldots,m$, and $j=1,\ldots,k$. This 1-form encodes the first law of thermodynamics
\begin{equation}
    dU = T dS + \sum_{i=1}^{m} X_i dY_i + \sum_{j=1}^k \mu_j  dN_j,
\end{equation}
which relates changes in internal energy to work and heat exchanges in equilibrium processes. $(U, S, Y_1, \ldots, Y_m, N_1, \ldots, N_k, T, X_1, \ldots, X_m, \mu_1, \ldots, \mu_k)$ span the thermodynamic state space, with $\eta = 0$ defining the equilibrium submanifold $E \subset M$, where thermodynamic states reside.

Following Callen \cite{callen85_thermodynamics} and Bravetti et al. \cite{bravetti15_contact}, we formulate classical thermodynamics through the following postulates, which provide a mathematically consistent foundation for the geometric approach.
\begin{postulate}\label{post:extensive}
\emph{Equilibrium states} of a thermodynamic system are fully characterized by a finite set of extensive variables: the internal energy $U$, generalized displacements $Y_1, \ldots, Y_m$, and particle numbers $N_1, \ldots, N_k$.
\end{postulate}

\begin{postulate}
\label{post:entropy}
There exists a differentiable function 
\begin{equation}
\begin{split}
	S: E & \longrightarrow \mathbb{R}\\
	(U, Y_1, \ldots, Y_m, N_1, \ldots, N_k) & \longmapsto S(U, Y_1, \ldots, Y_m, N_1, \ldots, N_k), 
\end{split}
\end{equation}
called the \emph{entropy}, defined on the space of equilibrium states $E$. This function is such that whenever an internal constraint of the system is removed, the new values of the extensive parameters are those which maximize the entropy over the manifold of constrained equilibrium states.
\end{postulate}

\begin{postulate}
\label{post:additivity}
For a composite system comprising multiple subsystems, the total entropy is the sum of the entropies of the subsystems. Furthermore, the entropy $S$ is a differentiable, monotonically increasing function of the internal energy $U$.
\end{postulate}

These postulates have immediate geometric and physical implications. 
First, the equilibrium states form a finite-dimensional submanifold $E \subset M$, parameterized by the extensive variables $(U, Y_1, \ldots, Y_m, N_1, \ldots, N_k)$. On the equilibrium submanifold $E$, the Gibbs 1-form satisfies $\eta = 0$, reflecting the first law in differential form. Finally, entropy is an intrinsic property for each particular subsystem and allows us to define the notion of \emph{temperature} as the relation between the entropy and the energy, more concretely $T:= (\partial S / \partial U)^{-1}$. This relation follows from Postulate \ref{post:entropy} and the monotonicity condition in Postulate \ref{post:additivity}.

This geometric perspective provides a robust foundation for extending thermodynamic concepts to quantum systems.
In classical thermodynamics, the thermodynamic state space and its contact structure capture the interplay between extensive and intensive variables. 
In the quantum case, we aim to generalize this structure to account for quantum states, described by density operators on a Hilbert space, and their thermodynamic properties. The equilibrium submanifold is replaced by a quantum state space, and the Gibbs 1-form will be adapted to incorporate quantum mechanical constraints.

The \emph{third law of thermodynamics}, which governs the asymptotic behavior of entropy as temperature approaches absolute zero, is not considered in this geometric analysis in the classical regime. However, in this quantum case, a version of the third law will arise naturally from the mathematical structure of the theory.

\section{Contact geometry}
To establish the geometrical framework for quantum thermodynamics, we first review the necessary concepts in contact geometry that underpin the thermodynamic state space and its extension to quantum systems. As stated, contact geometry provides a natural mathematical structure for describing thermodynamic equilibrium states and processes. This section introduces the essential definitions and results, setting the stage for their application to quantum thermodynamics in subsequent sections.

\begin{definition}
\label{def:contact_form}
Let $M$ be a smooth manifold of dimension $2n+1$. A 1-form $\eta \in \Lambda^1(M)$ is a \emph{contact form} if it satisfies the non-integrability condition
\begin{equation}\label{eqn:non_integrability}
    \eta \wedge (d\eta)^n \neq 0,
\end{equation}
where $(d\eta)^n$ is the $n$-fold exterior product of $d\eta$. This condition ensures that $\eta$ defines a \emph{volume form} on $M$.
\end{definition}
The non-integrability condition implies that contact forms exist only on odd-dimensional manifolds, as the wedge product $\eta \wedge (d\eta)^n$ is a $(2n+1)$-form, matching the dimension of $M$. 
\begin{proposition}
    The contact form induces a contact distribution
\begin{equation}
   \mathscr{D} = \ker(\eta) = \{ v \in TM \mid \eta(v) = 0 \}, 
\end{equation}
where $TM$ is the tangent bundle of $M$.
\end{proposition}
The distribution $\mathscr{D}$ has rank $2n$, and the condition $\eta \wedge (d\eta)^n \neq 0$ ensures that $\mathscr{D}$ is \emph{maximally non-integrable}. It follows from Frobenius' theorem that any submanifold of $M$ tangent to $\mathscr{D}$ has dimension at most $n$.

\begin{definition}
\label{def:contact_manifold}
A \emph{contact manifold} is a pair $(M, \eta)$, where $M$ is a $(2n+1)$-dimensional smooth manifold and $\eta$ is a contact form on $M$.
\end{definition}
In the context of thermodynamics, the thermodynamic state space is modeled as a contact manifold, with the Gibbs 1-form, Eq. (\ref{eqn:gibbs}), serving as the contact form. The contact structure encodes the thermodynamic relations, such as the first law.

 \begin{definition} An $n$-dimensional submanifold $\mathcal{L} \subset M$ of a contact manifold $(M, \eta)$ is \emph{Legendrian} if $\eta|_{\mathcal{L}}=0$.
 \end{definition}
In thermodynamics, Legendrian submanifolds correspond to equilibrium submanifolds $E \subset M$, where the Gibbs 1-form satisfies $\eta = 0$. These submanifolds represent the set of equilibrium states, parameterized by variables such as internal energy $U$, entropy $S$, volume $V$, and particle numbers $N_i$.

Since $\eta \wedge (d\eta)^n \neq 0$, the differential $d\eta$ must have one-dimensional kernel, which is transverse to  $\mathscr{D}$. This condition leads to the existence of a unique vector field on the manifold, which governs the dynamics.
 \begin{definition}
    The \emph{Reeb vector field} $\mathcal{R}_\eta$ on a contact manifold $(M, \eta)$ is the unique vector field satisfying the conditions $\eta(\mathcal{R_\eta})=1$, and  $d \eta(\mathcal{R_\eta}, \cdot) = 0$.
 \end{definition}
In thermodynamics, the Reeb field is associated with the evolution of thermodynamic processes. An essential result in contact geometry is the existence of canonical coordinates, which simplifies the local description of the contact structure.
\begin{theorem}[Darboux] \label{thm:darboux} Let $(M, \eta)$ be a $(2n+1)$-dimensional contact manifold. For any point $p \in M$, there exist local coordinates $(x_1, \ldots, x_n, y_1, \ldots, y_n, z)$ in a neighborhood of $p$ such that the contact form is expressed as
\begin{equation}
    \eta = dz - \sum_{i=1}^n y_i dx_i.
\end{equation}
\end{theorem} 
In thermodynamics, Darboux's theorem ensures that the thermodynamic state space can be locally described using coordinates that align the Gibbs 1-form with the contact structure. For example, for a simple system, coordinates can be chosen such that $\eta = dU - T dS + p dV$. This canonical form facilitates the analysis of thermodynamic relations and their extension to quantum systems.

The elements of contact geometry collectively provide a rigorous mathematical encoding of the intrinsic constraints and processes in thermodynamics. Physically speaking, the contact form $\eta$ embodies the non-integrable nature of the first law.
The contact distribution $\ker(\eta)$ corresponds to directions in the state space where reversible transformations can occur. The Reeb vector field $  \mathcal{R}_\eta$, which corresponds to $\partial_S$ in Darboux's coordinates, defines directions transverse to the equilibrium submanifolds.

\section{Connection to quantum thermodynamics}
In this section, we construct the quantum thermodynamic state space as a contact manifold, generalizing the classical thermodynamic state space to quantum systems. The quantum thermodynamic state space provides a geometric framework for describing quantum equilibrium states, encoded by Gibbs states, and their thermodynamic relations.

\begin{definition}
Let $\mathcal{H}$ be an $m$-dimensional complex Hilbert space, so that $\mathcal{H} \cong \mathbb{C}^m$ for $ m < \infty $, and let $ \mathscr{B}(\mathcal{H}) $ be the set of bounded operators on $\mathcal{H}$. A \emph{quantum state} is any $\rho \in \mathscr{B}(\mathcal{H})$ such that $\rho$ is self-adjoint, positive semidefinite, and has unit trace.
\end{definition}
See \cite{gustafson03_mathematical, takhtadzhian08_quantum} for further details. In finite dimensions, $ \mathscr{B}(\mathcal{H}) \cong M_m(\mathbb{C}) $. By virtue of this isomorphism, we identify states with unit trace self-adjoint matrices and introduce the following topological and differentiable structure.

\begin{proposition}
The set of quantum states $\mathcal{D} = \{ \rho \in \mathscr{B}(\mathcal{H}) \mid  \rho \ \mathrm{ quantum \ state}  \}$ is a compact convex subset of the hyperplane $\Pi = \{ \pi \in \mathscr{B}(\mathcal{H}) \mid \pi \ \mathrm{self}{-}\mathrm{adjoint},\ \tr(\pi) = 1 \} \cong \mathbb{R}^{m^2 - 1}$. 
\end{proposition}
\begin{proof}
Let $\rho_1, \rho_2 \in \mathcal{D}$ and $\lambda \in [0,1]$. Then $\lambda \rho_1 + (1 - \lambda) \rho_2$ is self-adjoint (as a convex combination of self-adjoint operators), positive semidefinite (since the positive semidefinite cone is convex: the eigenvalues of the combination lie in the convex hull of the eigenvalues of $\rho_1$ and $\rho_2$, hence remain non-negative), and satisfies $\tr[\lambda \rho_1 + (1 - \lambda) \rho_2] = 1$ by linearity of the trace. Thus, $\mathcal{D}$ is convex.
\end{proof}
\begin{definition}
    The \emph{relative interior} of $\mathcal{D}$ is the open set of \emph{positive definite} states, defined as $\mathcal{D}^\circ := \operatorname{int}_{\Pi}(\mathcal{D}) = \{ \rho \in \Pi \mid \rho > 0 \}$.
\end{definition}
\begin{remark}
    The positive definite states are those with full rank, i.e., $\mathrm{rank}(\rho) = m$, for all $\rho \in \mathcal{D}^\circ$.
\end{remark}
\begin{proposition}
    $\mathcal{D}^\circ$ is a smooth manifold of dimension $m^2 - 1$.
\end{proposition}
\begin{proof}
    The manifold structure of $\mathcal{D}^\circ$ is trivial \cite{caja12_variedades}.
\end{proof}

\begin{definition}
The \emph{boundary} of the state space is $\partial \mathcal{D} := \mathcal{D} \setminus \mathcal{D}^\circ = \{ \rho \in \mathcal{D} \mid \mathrm{rank}(\rho) < m \}$.
\end{definition}
\begin{remark}
   The boundary $\partial \mathcal{D}$ contains all rank-deficient states, i.e., $\mathrm{rank}(\rho) < m$, including the \emph{pure states}, i.e., states with rank $1$. 
\end{remark}

The subsequent analysis is restricted to the interior $\mathcal{D}^\circ$. Although boundary effects must be considered in general equilibrium manifold theory due to finite-dimensional constraints, they play no role in the present context, as will be shown in Sec.~\ref{sec:boundary}.

\begin{definition}
Let $\mathcal{H}$ be a Hilbert space. An \emph{observable} $A$ is a self-adjoint operator on $\mathcal{H}$.
For any quantum state $\rho \in \mathcal{D}$, the \emph{expectation value} of $A$ in the state $\rho$ is defined as
\begin{equation}
a = \langle A \rangle_\rho = \tr(\rho A) \in \mathbb{R}.
\end{equation}
A finite collection $\{A_1, \ldots, A_n\}$ of observables is called a \emph{set of observables}.
\end{definition}

For later convenience, we introduce a finite set of observables $ \boldsymbol{A} = \{ A_i \in \mathscr{B}(\mathcal{H}) \mid A_i \text{ observable}, i=1,\ldots,n \}$.
We associate with this set a collection of real-valued \emph{conjugate variables} $\boldsymbol{\lambda} = \{ \lambda_i \in \mathbb{R} \mid i=1,\ldots,n \}$, acting as smooth coordinates on the thermodynamic state space. Mimicking classical thermodynamics, the interest resides in defining a contact manifold for the quantum states, for them to be labeled by several coordinates by Theorem \ref{thm:darboux}. 
\begin{theorem}[Quantum thermodynamic state space]
\label{thm:qtps}
The \emph{quantum thermodynamic state space} is a smooth manifold $\mathcal{M} \cong \mathbb{R}^{2n+1}$, with global coordinates $(S,\boldsymbol{a},\boldsymbol{\lambda}) = (S, a_1, \ldots, a_n, \lambda_1, \ldots, \lambda_n)$, endowed with a contact form
\begin{equation}
  \eta = dS - \sum_{i=1}^n \lambda_i da_i \in \Lambda^1(\mathcal{M}).  
\end{equation}
\end{theorem}
\begin{proof}
The diffeomorphism $\mathcal{M}\cong\mathbb{R}^{2n+1}$ is immediate from the coordinate chart.
To verify that $\eta$ is a contact form, compute the exterior derivative
\begin{equation}
    d\eta = d\left(dS - \sum_{i=1}^n \lambda_i da_i\right) = -\sum_{i=1}^n d\lambda_i \wedge da_i,
\end{equation}
since $d(dS) = 0$ and $d(\lambda_i da_i) = d\lambda_i \wedge da_i$. Then
\begin{equation}
        (d\eta)^{ n} = (-1)^n \bigwedge_{i=1}^n \left(d\lambda_i \wedge da_i\right).
    \end{equation}
and the wedge product is
\begin{equation}
    \eta \wedge (d\eta)^n = dS \wedge (-1)^n \bigwedge_{j=1}^n (d\lambda_i \wedge da_i).
    \end{equation}
The set $\{dS,da_1,\dots,da_n,d\lambda_1,\dots,d\lambda_n\}$ is a basis of $T^*\mathcal{M}$, so the right-hand side is a non-vanishing volume form. Hence $\eta\wedge (d\eta)^n\neq 0$ everywhere and $(\mathcal{M},\eta)$ is a contact manifold.
\end{proof}
To characterize quantum states in thermodynamic equilibrium, we introduce the following postulate, which specifies their mathematical form.
\begin{postulate}\label{post:thermal_state}
    A quantum state $\rho\in \mathcal{D}$ is in \emph{thermodynamic equilibrium}  with respect to a set of observables $\{A_1,\dots,A_n\}\subset\mathscr{B}(\mathcal{H})$ if it is given by
    \begin{equation}
\rho_{\boldsymbol{\lambda}}
 = \frac{\exp\left(-\sum_{i=1}^n\lambda_i A_i\right)}{Z(\boldsymbol{\lambda})},
\label{eq:gibbs}
\end{equation}
with
\begin{equation}
Z(\boldsymbol{\lambda})=\tr \left[\exp\left(-\sum_{i=1}^n\lambda_i A_i\right)\right],
\end{equation}
where $\boldsymbol{\lambda}=(\lambda_1,\dots,\lambda_n)\in \mathbb{R}^n$. Such states are known as \emph{Gibbs states}, and maximize the entropy for a given set of observables $\{A_1,\dots,A_n\}\subset\mathscr{B}(\mathcal{H})$.
\end{postulate}
\begin{remark}
	Since the exponential of a matrix is always invertible, it has maximum rank, meaning that all states of the form given by Eq. (\ref{eq:gibbs}) lie in $\mathcal{D}^\circ$.
\end{remark}
To establish a connection between the contact structure and the quantum viewpoint, we define a state function that maps thermodynamic coordinates to quantum states.

\begin{definition}
\label{def:state_function}
Let $\mathcal{D}$ be the set of all density operators on a finite-dimensional Hilbert space $\mathcal{H}\cong\mathbb{C}^m$.  
The \emph{state function} $\Xi: \mathcal{M} \to \mathcal{D}$ assigns to each thermodynamic point $(S,\boldsymbol{a},\boldsymbol{\lambda})$ the density operator
\begin{equation}
\Xi(S,\boldsymbol{a},\boldsymbol{\lambda})
 = \frac{\exp\left(-\sum_{i=1}^n\mu_i(S,\boldsymbol{a},\boldsymbol{\lambda})A_i\right)}
       {\tr\left[\exp\left(-\sum_{i=1}^n\mu_i(S,\boldsymbol{a},\boldsymbol{\lambda})A_i\right)\right]},
\label{eq:state-function}
\end{equation}
where the smooth functions $\mu_i:\mathcal{M}\to\mathbb{R}$ satisfy the \emph{equilibrium consistency condition}, i.e., given the embedding
\begin{equation}
\begin{split}
	\iota: \mathbb{R}^n & \longrightarrow\mathcal{M}\\
	\boldsymbol{\lambda} &\longmapsto\left(S(\boldsymbol{\lambda}),\boldsymbol{a}(\boldsymbol{\lambda}),\boldsymbol{\lambda}\right), 
\end{split}
\end{equation}
where
\begin{equation}
\begin{split}
	a_i(\boldsymbol{\lambda})& =\langle A_i\rangle_{\rho_{\boldsymbol{\lambda}}}=-\frac{\partial\log Z}{\partial\lambda_i},\\
	S(\boldsymbol{\lambda})&=-\tr\left(\rho_{\boldsymbol{\lambda}}\log\rho_{\boldsymbol{\lambda}}\right)
	=\log Z(\boldsymbol{\lambda})+\sum_{i=1}^n\lambda_i a_i(\boldsymbol{\lambda}),
\end{split}
\end{equation}
we have
\begin{equation}
(\mu_i \circ \iota)(\boldsymbol{\lambda})=\lambda_i,
\label{eq:mu-equilibrium}
\end{equation}
for all $\boldsymbol{\lambda}\in \mathbb{R}^n$, for $i = 1, \ldots, n$. Additionally, for fixed values of $S$ and $\boldsymbol{a}$, and also $c_1, \dots, c_n \in \mathbb{R}$, the system
\begin{equation}
	\mu_1(S,\boldsymbol{a},\boldsymbol{\lambda}) = c_1, \hspace{3pc} \hdots \hspace{3pc} 
	\mu_n(S,\boldsymbol{a},\boldsymbol{\lambda}) = c_n,
\end{equation}
has a unique solution $\lambda_1, \dots, \lambda_n \in \mathbb{R}$. This condition will be called \textit{joint injectivity} of the $\mu_i$ functions.

\end{definition}
The functions $\mu_i$ generalize the conjugate parameters $\lambda_i$ to non-equilibrium points, ensuring that $\Xi$ is a smooth map across $\mathcal{M}$.
On the equilibrium submanifold, where $\mu_i = \lambda_i$, the state reduces to a Gibbs state, which maximizes the entropy for given expectation values $a_i = \tr(A_i \rho)$.

\begin{proposition}
\label{prop:equilibrium}
The image of the previous embedding is an $n$-dimensional \emph{Legendrian submanifold} $\mathcal{E}=\iota(\mathbb{R}^n)\subset\mathcal{M}$.
\end{proposition}

\begin{proof}
On $\mathcal{E}$ the functions $\mu_i$ reduce to $\lambda_i$ by Eq. \eqref{eq:mu-equilibrium}, so
\begin{equation}
\Xi\bigl(S(\boldsymbol{\lambda}),\boldsymbol{a}(\boldsymbol{\lambda}),\boldsymbol{\lambda}\bigr)=\rho_{\boldsymbol{\lambda}}.
    \end{equation}
Differentiating the identity $S=\log Z+\sum_i\lambda_i a_i$
\begin{equation}
dS = \sum_i\lambda_ida_i + \sum_i a_id\lambda_i
          + \sum_i\frac{\partial\log Z}{\partial\lambda_i}d\lambda_i.
    \end{equation}
The partial derivative $\partial\log Z/\partial\lambda_i=-a_i$, hence the last two terms cancel and $dS = \sum_{i=1}^n\lambda_ida_i$.
The pull-back of the contact form is therefore $\iota^*\eta = dS - \sum_i\lambda_ida_i = 0$.
Since $\dim\mathcal{E}=n$ and $\iota^*\eta=0$, $\mathcal{E}$ is Legendrian.
\end{proof}

In thermodynamic equilibrium, i.e., on $\mathcal{E}$, the coordinates have physical interpretations. 
The coordinate $S$ is the \emph{von Neumann entropy} $S=-\tr(\rho\log\rho)$.  
Each $a_i=\tr(A_i\rho)$ is the \emph{expectation value} of the observable $A_i$.  
The conjugate parameters $\lambda_i$ play the role of \emph{intensive variables}, e.g., inverse temperatures or chemical potentials. The contact relation $dS=\sum_i\lambda_ida_i$ is the quantum version of the classical first law in the entropy representation.

\begin{remark}[First law of quantum thermodynamics]
\label{remark:quantum_first_law}
The quantum first law is encoded in the contact form $\eta = dS - \sum_{i=1}^n \lambda_i da_i$ on the quantum thermodynamic state space $\mathcal{M}$. On the Legendrian submanifold $\mathcal{E}$ of equilibrium states, $\eta|_{\mathcal{E}} = 0$ implies
\begin{equation}
dS = \sum_{i=1}^n \lambda_i da_i,
\end{equation}
expressing the differential change of the entropy in terms of variations in expectation values of observables and intensive parameters. This generalizes the classical first law to quantum systems.
\end{remark}

\begin{remark}
Henceforth, thermodynamic entropy in the quantum regime is quantified by the \emph{von Neumann entropy}, defined as $ S(\rho) = -\tr(\rho \log \rho)$, for all $\rho \in \mathcal{D}^\circ$.
Although the precise interpretation of quantum thermodynamic entropy remains a subject of active debate, for our purposes, the specific interpretation is immaterial, provided the entropy is a smooth, positive, and concave function $S: \mathcal{D}^\circ \to \mathbb{R}$ on the manifold of density operators. These conditions ensure compatibility with the differential geometric framework developed herein.
\end{remark}

The functions $\mu_i: \mathcal{M} \to \mathbb{R}$, $i=1,\dots,n$ allow $\Xi$ to be defined smoothly everywhere in $\mathcal{M}$, generalizing Gibbs states to non-equilibrium points.
On the \emph{equilibrium submanifold} $\mathcal{E}$, the consistency condition given by Eq. \eqref{eq:mu-equilibrium} enforces $\mu_i\left(S(\boldsymbol{\lambda}), \boldsymbol{a}(\boldsymbol{\lambda}), \boldsymbol{\lambda}\right) = \lambda_i$, for all $i = 1, \ldots,n$
so that $\Xi\bigl(S(\boldsymbol{\lambda}), \boldsymbol{a}(\boldsymbol{\lambda}), \boldsymbol{\lambda}\bigr) = \rho_{\boldsymbol{\lambda}}$
recovers the standard Gibbs state as in Eq. \eqref{eq:gibbs}.
Off equilibrium, i.e., for points away from $\mathcal{E}$, the functions $\mu_i$ are \emph{not} constrained by the Gibbs form and may be chosen flexibly to model non-equilibrium dynamics. For example, the correspondence $\mu_i = \lambda_i + f_i(S, \boldsymbol{a}, \boldsymbol{\lambda})$,
where each $f_i \in C^\infty(\mathcal{M})$ vanishes on $\mathcal{E}$, i.e., $f_i|_\mathcal{E} = 0$, provides a valid extension, for all $i = 1, \ldots,n$. 

A naive attempt to define $a_i = \tr(A_i \Xi(S, \boldsymbol{a}, \boldsymbol{\lambda}))$ globally on $\mathcal{M}$ would render $\Xi$ implicitly dependent on its own output, leading to circularity. This is avoided by treating the coordinates $(S, \boldsymbol{a}, \boldsymbol{\lambda})$ as \emph{independent} in $\mathcal{M} \cong \mathbb{R}^{2n+1}$. The physical constraint by $a_i = \tr(A_i \rho)$
is imposed only on the equilibrium submanifold $\mathcal{E}$, where it is automatically satisfied by the Gibbs construction: see Proposition~\ref{prop:equilibrium}. Elsewhere, $\boldsymbol{a}$ represents \emph{target} or \emph{virtual} expectation values, not necessarily realized by $\Xi$, enabling the modeling of constrained or fictitious ensembles.

\section{Quantum thermodynamic fiber bundle}
Under appropriate circumstances, the state function introduced in Definition~\ref{def:state_function} induces a rich geometric structure: $\mathcal{M}$ is a \emph{fiber bundle} over the base space of quantum states, with fibers encoding all thermodynamic configurations compatible with a given density operator. This framework unifies equilibrium and non-equilibrium quantum thermodynamics within a single differential-geometric object.

Let us suppose that the manifold of Gibbs states generated by the fixed observables $\{A_1, \dots, A_n\}$, which is given by $\mathcal{B} = \bigl\{ \rho_{\boldsymbol{\lambda}} \big| \boldsymbol{\lambda} \in U \subseteq \mathbb{R}^n\bigr\} \subseteq \mathcal{D}^\circ$, has maximal dimension $n$. Then, the correspondence $\Xi$ is a smooth map onto the $n$-dimensional submanifold $\mathcal{B}$. Additionally, if $\Xi$ is a surjective \emph{submersion}, i.e., the differential $ d\Xi_p: T_p \mathcal{M} \to T_{\Xi(p)} \mathcal{B} $ is surjective for all $ p \in \mathcal{M} $, we can claim the following:

\begin{definition}
Let $\mathcal{M} \cong \mathbb{R}^{2n+1}$ be the quantum thermodynamic state space with coordinates $(S, \boldsymbol{a}, \boldsymbol{\lambda})$, and let $\mathcal{B}$ be the manifold of Gibbs states generated by the fixed observables $\{A_1, \dots, A_n\}$. The \emph{quantum thermodynamic fiber} over a Gibbs state $\sigma \in \mathcal{B}$ is the preimage
\begin{equation}
F_\sigma := \Xi^{-1}(\sigma) = \bigl\{ (S, \boldsymbol{a}, \boldsymbol{\lambda}) \in \mathcal{M} \;\big|\; \Xi(S, \boldsymbol{a}, \boldsymbol{\lambda}) = \sigma \bigr\}.
\label{eq:fiber}
\end{equation}
If $ \Xi $ is a submersion, the fiber $F_\sigma$ is a smooth submanifold of $ \mathcal{M} $ with dimension
\begin{equation}
	\dim F_\sigma = \dim \mathcal{M} - \dim \mathcal{B} = (2n + 1) - n = n + 1.
	\label{eq:fiber_dim}
\end{equation}
\end{definition}

\begin{remark}
	The conditions required for the previous definition to make sense, namely that the submanifold $\mathcal{B}$ has dimension $n$ and that the map $\Xi$ be a submersion, are not guaranteed to hold. A first immediate requirement for $\mathcal{B}$ to have maximal dimension is that the observables $\{A_1, \dots, A_n\}$ be linearly independent. This ensures that the submanifold generated by the exponential of the linear combinations of such observables has the correct dimension $n$. However, dividing by the trace may entail a reduction in the dimensionality of the submanifold. For example, if the set of observables only has one element $\{ A_1 \}$ and $A_1$ is a multiple of the identity matrix, the set $\mathcal{B}$ contains just one element. On the contrary, if $A_1$ is an $n \times n$ matrix having $n$ different eigenvalues, $\mathcal{B}$ has dimension 1 and $\Xi$ is a submersion. Cases where the set of observables has more than one element have to be studied separately.
\end{remark}

If all the level sets $\Xi^{-1}(\sigma)$ for $\sigma \in \mathcal{B}$ are diffeomorphic, the previous construction gives rise to a fiber bundle:
\begin{definition}
A \emph{quantum thermodynamic fiber bundle} is a tuple $ (\mathcal{M}, \mathcal{B}, \Xi, F)$, where 
$\mathcal{M} $ is the quantum thermodynamic state space (total space), with $\dim = 2n+1$;
$\mathcal{B} \subseteq \mathcal{D}^\circ$ is the base space, with $\dim = n$;
$F$ is the \emph{fiber}; 
and $\Xi: \mathcal{M} \to \mathcal{B}$ is the smooth \emph{projection} (state function).
The bundle satisfies \emph{local triviality}: for every $\sigma \in \mathcal{B}$, there exists a neighborhood $U \subset \mathcal{B}$ and a diffeomorphism $\phi: \Xi^{-1}(U) \to U \times F$, such that the following diagram commutes
\begin{equation*}
\begin{tikzcd}
\Xi^{-1}(U) \arrow[r, "\phi"] \arrow[dr, "\Xi"'] & U \times F \arrow[d, "\mathrm{pr}_1"] \\
& U
\end{tikzcd}
\end{equation*}
where $ \mathrm{pr}_1: U \times F \to U $ is the projection onto the first factor. Thus, $(\mathcal{M}, \Xi, \mathcal{B})$ is a \emph{smooth fiber bundle} with fiber $F$.
\end{definition}

The state $\Xi(S, \boldsymbol{a}, \boldsymbol{\lambda})$ depends \emph{only} on the $n$ values $\boldsymbol{\mu}(S, \boldsymbol{a}, \boldsymbol{\lambda}) = (\mu_1, \dots, \mu_n) \in \mathbb{R}^n$. Hence, $\sigma = \Xi(S, \boldsymbol{a}, \boldsymbol{\lambda})$ if and only if $\mu_i(S, \boldsymbol{a}, \boldsymbol{\lambda}) = c_i$, where $\boldsymbol{c} = (c_1, \dots, c_n)$ determines $\sigma = \rho_{\boldsymbol{c}}$ via the Gibbs form.  
The fiber $F_\sigma$ is therefore the \emph{joint level set}
\begin{equation}
F_\sigma = \bigl\{ (S, \boldsymbol{a}, \boldsymbol{\lambda}) \in \mathcal{M} \;\big|\; \mu_i(S, \boldsymbol{a}, \boldsymbol{\lambda}) = c_i , i=1,\ldots,n \bigr\},
\end{equation}
a smooth $(n+1)$-dimensional submanifold under the assumptions stated earlier.

\begin{remark}
	Under the conditions thus far assumed, the quantum thermodynamic fiber bundle is a \textit{trivial bundle}, in the sense that it is globally diffeomorphic to a cartesian product $\mathbb{R}^n \times \mathbb{R}^{n+1}$. Indeed, the base space of the bundle, namely the submanifold of Gibbs states $\mathcal{B}$, is always contractible for a quantum system of finite-dimensional state space. It can be seen that $\mathcal{B}$ is homeomorphic to $\mathbb{R}^n$. As is well known, any fiber bundle over a contractible base space is necessarily trivial. In other words, $\mathbb{R}^{2n+1}$ is homeomorphic to $\mathbb{R}^n \times \Xi^{-1}(\sigma)$ for any $\sigma \in \mathcal{B}$. From topological considerations, it then follows that $\Xi^{-1}(\sigma)$ must also be contractible. In fact, $\Xi^{-1}(\sigma)$ can be argued to be diffeomorphic to $\mathbb{R}^{n+1}$.
\end{remark}

The equilibrium submanifold $\mathcal{E} \subset \mathcal{M}$ is a \emph{Legendrian section} of the bundle. For $\sigma \in \mathcal{B}$, the intersection $F_\sigma \cap \mathcal{E}$
consists of a single element, assuming the map $\boldsymbol{\lambda} \mapsto \rho_{\boldsymbol{\lambda}}$ is injective, corresponding to the unique thermodynamic coordinates $\left(S(\boldsymbol{\lambda}), \boldsymbol{a}(\boldsymbol{\lambda}), \boldsymbol{\lambda}\right)$,
where $\mu_i = \lambda_i$, $S = -\tr(\rho \log \rho)$, and $\boldsymbol{a} = \langle \boldsymbol{A} \rangle_\rho$. This point satisfies the contact constraint $\eta = 0$, i.e.,
\begin{equation}
dS = \sum_{i=1}^n \lambda_ida_i.
\end{equation}
Points in $F_\sigma \setminus \mathcal{E}$ represent \emph{non-equilibrium thermodynamic configurations} that still yield the same physical state $\sigma$ but violate equilibrium relations, e.g., inconsistent $\boldsymbol{a}$.

The physical interpretation is that the fiber $F_\sigma$ is the set of all thermodynamic labels $(S, \boldsymbol{a}, \boldsymbol{\lambda})$ compatible with the same physical density matrix, $\sigma$.  
On $\mathcal{E}$, we have \emph{equilibrium}, i.e., unique $S$, correct $\boldsymbol{a} = \langle \boldsymbol{A} \rangle_\sigma$, $\boldsymbol{\lambda}$ as intensive parameters. Off $\mathcal{E}$, the points represent \emph{non-equilibrium}, meaning same $\sigma$, but $S \neq -\tr(\sigma \log \sigma)$, or $\boldsymbol{a} \neq \langle \boldsymbol{A} \rangle_\sigma$

\begin{example}
\label{example:qubit_fiber}
To build intuition, consider a \emph{qubit}, with $\mathcal{H} \cong \mathbb{C}^2$, so $m=2$, with a single observable $A_1 = \sigma_z$, the Pauli-$Z$ matrix. This is given by
      \begin{equation}
	\sigma_z = \begin{pmatrix} 1 & 0 \\ 0 & -1 \end{pmatrix}.
\end{equation}
Then, $\mathcal{M} \cong \mathbb{R}^3$ with coordinates $(S, a, \lambda)$; $\mathcal{B} \subset \mathcal{D}^\circ$ is the \emph{thermal curve} of Gibbs states
      \begin{equation}
      \rho_\lambda = \frac{1}{2}\begin{pmatrix} 1 + \tanh\lambda & 0 \\ 0 & 1 - \tanh\lambda \end{pmatrix},
      \end{equation}
      parametrized by inverse temperature $\lambda \in \mathbb{R}$; and the expectation value of the observable and the entropy are
      \begin{equation}
          \begin{split}
              a( \lambda) &= \langle \sigma_z \rangle = \tanh\lambda, \\
              S(\lambda) &= \log(2\cosh \lambda) - \lambda \tanh\lambda.
          \end{split}
      \end{equation}
The state function $\Xi: \mathbb{R}^3 \to \mathcal{B}$ assigns to each thermodynamic point $(S, a, \lambda)$ the density matrix
\begin{equation}
\Xi(S, a, \lambda) = \frac{\exp(-\mu(S, a, \lambda) \sigma_z)}{\tr[\exp(-\mu(S, a, \lambda) \sigma_z)]},
\end{equation}
where $\mu: \mathbb{R}^3 \to \mathbb{R}$ is any smooth extension of the equilibrium relation $\mu(S(\lambda), a(\lambda), \lambda) = \lambda$ satisfying injectivity over the $\lambda$ variable.

Geometrically, this means that $\mathcal{M} \cong \mathbb{R}^3$ is a 3D volume; $\mathcal{B} \subset \mathcal{D}^\circ$ is a 1D curve in the Bloch ball interior; the fiber $F_\sigma = \Xi^{-1}(\sigma)$ over a fixed thermal state $\sigma = \rho_{\lambda_0}$ is the set of all $(S, a, \lambda)$ such that $\mu(S, a, \lambda) = \lambda_0$, i.e., a 2D surface (codimension 1) in $\mathbb{R}^3$.

The \emph{equilibrium submanifold} $\mathcal{E}$ is the curve $
\mathcal{E} = \{ (S(\lambda), a(\lambda), \lambda) \mid \lambda \in \mathbb{R}\}$,
which intersects each fiber $F_\sigma$ at \emph{exactly one point}---the unique thermodynamic coordinates consistent with thermal equilibrium. Thus, multiple points in $\mathcal{M}$ map to the same physical state $\sigma$, but only one lies on $\mathcal{E}$. Points off $\mathcal{E}$ in $F_\sigma$ represent non-equilibrium thermodynamic descriptions of $\sigma$: same density matrix, but incorrect entropy $S$, mismatched expectation $a$, or inconsistent intensive parameter $\lambda$. This redundancy is the hallmark of the fiber bundle structure. This fact is clearer in the following sections.
\end{example}

\begin{remark}[Physical implications of the injectivity and the zeroth law of quantum thermodynamics]
\label{remark:gibbs_injectivity_zeroth}
The zeroth law---the transitivity of thermal equilibrium---rests on the injectivity of the Gibbs map $\boldsymbol{\lambda} \mapsto \rho_{\boldsymbol{\lambda}}$.

In the quantum thermodynamic fiber bundle, this injectivity ensures that each Gibbs state $\sigma \in \mathcal{B}$ intersects the equilibrium submanifold $\mathcal{E}$ at exactly one point.
Thermal equilibrium between $A$ and $B$ therefore corresponds to both systems mapping to the same fiber $F_{\rho_A = \rho_B}$, and thus to the same unique equilibrium point $p_\sigma \in \mathcal{E}$. The shared coordinates $\boldsymbol{\lambda}$ define a universal intensive parameter field on $\mathcal{E}$, with level sets, e.g., fixed inverse temperature $\beta$, forming transitive equilibrium classes.
The zeroth law is then seen to be a consequence of the geometric structure:  thermal equilibrium is transitive because equilibrium is unique.

If injectivity fails, by e.g., redundant observables $A_2 = c A_1$, with $c \in \mathbb{R}$, multiple $\boldsymbol{\lambda}$ provide the same $\rho$, destroying uniqueness and allowing inconsistent temperature assignments, the zeroth law breaks. Thus, injectivity is not an additional assumption but the geometric expression of thermodynamic uniqueness: a single quantum state admits one and only one equilibrium thermodynamic description. This structural fact is what allows temperature (and all intensive parameters) to be unambiguously defined, shared across systems in contact, and transitive under thermal equilibration---the very essence of the quantum zeroth law.
\end{remark}

	\begin{remark}
Under the assumption that the Gibbs map $\boldsymbol{\lambda} \mapsto \rho_{\boldsymbol{\lambda}}$ is injective (as required for the zeroth law, see Remark \ref{remark:gibbs_injectivity_zeroth}), the equilibrium submanifold $\mathcal{E}$ can be identified diffeomorphically with the manifold of Gibbs states $\mathcal{B}$. This diffeomorphism is explicitly realized by the map $\Phi : \mathcal{E} \longrightarrow \mathcal{B}$ given by $\rho_{\boldsymbol{\lambda}} \circ \pi$, where $\pi$ denotes the projection $(S, \boldsymbol{a}, \boldsymbol{\lambda}) \longmapsto \boldsymbol{\lambda}$. As a consequence of this identification, we may understand the base space of the quantum thermodynamic fiber bundle to be either $\mathcal{B}$ or $\mathcal{E}$, depending on what is most suitable for the issue being discussed.

By smoothly extending the functions $\mu_i$ off $\mathcal{E}$, the state function $\Xi$ becomes well-defined everywhere, with the property that its image is precisely the submanifold $\mathcal{B}$.
In this extended construction, $\mathcal{B}$ appears as the distinguished $n$-dimensional submanifold of thermal equilibrium states inside $\mathcal{D}^\circ \supset \mathcal{B}$, and $\mathcal{E} \cong_{\Phi} \mathcal{B} \subset \mathcal{D}^\circ$. 
For every arbitrary state $\rho \in \mathcal{D}^\circ$ associate a unique effective thermal state $\Phi(p_\rho) \in \mathcal{B}$, where $p_\rho \in \mathcal{E}$ is the equilibrium point sharing the same expectation values $\boldsymbol{a}(\rho)$ and entropy $S(\rho)$. This fact enables a consistent fiber-bundle description for every full-rank density operator in $\mathcal{D}^\circ$ through its projection onto the Gibbs submanifold $\mathcal{B}$.
  \end{remark}

\begin{example}
\label{example:nonthermal_qubit}
Consider again the qubit with $\mathcal{H} \cong \mathbb{C}^2$ and observable $A = \sigma_z$. Take a non-thermal full-rank state, e.g., the diagonal density operator $\rho = \mathrm{diag} (0.9,  0.1)$, with expectation value $a = \tr(\rho \sigma_z) = 0.8$ and von Neumann entropy $S(\rho) = 0.325$.

If we choose $\mu(S, a, \lambda) = \lambda$, the unique equilibrium intersection is the thermal Gibbs state with the same $a=0.8$, i.e., $\rho_{\lambda} = \mathrm{diag}((1 + \tanh\lambda)/2, (1 - \tanh\lambda)/2)$ where $\tanh\lambda = 0.8$, so $\lambda = 1.099$. The equilibrium labels are $S_{\mathrm{eq}} = 0.420 > S(\rho)$ (maximum entropy for fixed $a$), $a_{\mathrm{eq}} = 0.8$, $\lambda_{\mathrm{eq}} = 1.099$. Points off in $F_{\rho_\lambda}$ correspond to alternative thermodynamic descriptions of this same non-thermal $\rho$.

Physically, this non-thermal $\rho$ has excess free energy relative to its passive thermal counterpart: work can be extracted unitarily to reach the reordered passive state, then further to the ground. The fiber geometry quantifies distance to equilibrium, e.g., via pseudo-metric lengths for relaxation paths projecting to thermalization under a bath.
A horizontal lift of a curve from $\rho$ to $\rho_{\lambda_{\mathrm{eq}}}$ in the base represents an optimal quasistatic driving toward the effective thermal state.
\end{example}

\section{Movement along fibers}
The fibers $F_\sigma = \Xi^{-1}(\sigma)$ of the quantum thermodynamic fiber bundle provide a natural arena for analyzing thermodynamic processes at a fixed quantum state $\sigma \in \mathcal{B}$.
A path $\gamma: [0,T] \to \mathcal{M}$ confined to $F_\sigma$ evolves the thermodynamic coordinates $(S, \boldsymbol{a}, \boldsymbol{\lambda})$ while preserving the density operator $\Xi[\gamma(t)] = \sigma$, for all $t \in [0,T]$.

\begin{definition}
A smooth path $\gamma: [0,T] \to \mathcal{M}$ is state-preserving with respect to $\sigma \in \mathcal{B}$ if $\gamma(t) \in F_\sigma$, for all $t \in [0,T]$,
i.e., $\mu_i\big(S(t), \boldsymbol{a}(t), \boldsymbol{\lambda}(t)\bigr) = c_i$, for $i=1,\dots,n$, where $\boldsymbol{c} \in \mathbb{R}^n$ satisfies $\sigma = \rho_{\boldsymbol{c}}$ via the Gibbs form.
\end{definition}
The fiber $F_\sigma$ is a smooth $(n+1)$-dimensional submanifold of $\mathcal{M}$, defined by $n$ independent level set constraints.
Assuming injectivity of $\boldsymbol{\lambda} \mapsto \rho_{\boldsymbol{\lambda}}$, the equilibrium point in $F_\sigma$ is the unique intersection $
p_\sigma := F_\sigma \cap \mathcal{B}$,
where $\mu_i = \lambda_i$, $a_i = \tr(A_i \sigma)$, and $S = -\tr(\sigma \log \sigma)$. A relaxation process to the equilibrium is a path $\gamma: [0,\infty) \to F_\sigma$ with
the initial condition $\gamma(0) \in F_\sigma \setminus \mathcal{B}$, and limit $ \lim_{t \to \infty} \gamma(t) = p_\sigma$.
The contact distribution $\ker \eta \subset T\mathcal{M}$ defines Legendre submanifolds of dimension $n$. While $F_\sigma$ is $(n+1)$-dimensional, its intersection with $\ker \eta$ gives reversible directions. 

Paths $\gamma(t)\in F_\sigma$ keep the density operator $\sigma$ fixed, so the expectation values $a_i(t)=\tr(A_i\sigma)$ are target coordinates that may differ from the actual $\tr(A_i\sigma)$ except on the equilibrium point $p_\sigma=F_\sigma\cap\mathcal{B}$.
Thermodynamically, these paths in $F_\sigma$ enable computation of quantities like work or heat. The quantum first law can be expressed through changes in the expectation values, constrained by the structure of the fiber. For a quasistatic process, the work done is related to changes in $\boldsymbol{\lambda}(t)$, while heat is associated with $S'(t)$. The contact distribution $\ker\eta\subset T\mathcal{M}$ selects the reversible directions inside $F_\sigma$. Any deviation from $\ker\eta$ generates positive entropy production, providing a geometric criterion for thermodynamic reversibility. Note that if $\gamma(t) \in \ker \eta$, the process is reversible since there is no entropy production associated with such a process.

\section{Distances}
The preceding discussion on relaxation processes within the fiber suggests that this geometric formulation can quantify the separation between quantum states, whether in equilibrium or non-equilibrium, to measure their distance from the equilibrium configuration on the Legendrian submanifold $\mathcal{E}$.

In quantum theory, the \emph{Bures-Wasserstein distance} defines a distance measure comparing quantum states, represented by density matrices, by considering the properties of positive-definite self-adjoint matrices \cite{oostrum22_bures}. This distance is particularly suited for the quantum thermodynamic state space, as it aligns with the Riemannian geometry of the space of quantum states.
\begin{definition}
\label{def:bures_wasserstein}
Let $A, B$ be positive-definite self-adjoint matrices on a Hilbert space $\mathcal{H}$. The \emph{Bures-Wasserstein distance }is defined as
\begin{equation}
        d_{\mathrm{BW}}(A,B) = \left\{\tr(A) + \tr(B) - 2 \tr \left[\left(A^{1/2}BA^{1/2} \right)^{1/2}\right] \right\}^{1/2}.
    \end{equation}
    \end{definition}
Note that, if $\tr(A) = \tr(B) =1$, $ d_{\mathrm{BW}}(A,B)$ simplifies to 
    \begin{equation}
        d_{\mathrm{BW}}(A,B) = \left\{2 - 2 \tr \left[\left(A^{1/2}BA^{1/2} \right)^{1/2}\right] \right\}^{1/2},
    \end{equation}
where we can define $F(A,B) := \tr \left[ \sqrt{A^{1/2} B A^{1/2}} \right] $ the \emph{fidelity} between matrices $A$ and $B$. The Bures-Wasserstein distance is a Riemannian distance, inducing a Riemannian metric on the space of quantum states.

\begin{definition}
For a tangent vector $X \in T_{\rho_{\boldsymbol{\lambda}_0}} \mathcal{B}$ at a Gibbs state $\rho_{\boldsymbol{\lambda}_0} \in \mathcal{B}$, the directional derivative of the state map is given by
\begin{equation}
    X(\rho_{\boldsymbol{\lambda}}) = \left. \frac{d}{dt} \right|_{t=0} (\rho_{\boldsymbol{\lambda}} \circ \gamma)(t),
\end{equation}
where $\gamma : (-\varepsilon, \varepsilon) \longrightarrow \mathcal{B}$ is a smooth curve such that $\gamma(0) = \rho_{\boldsymbol{\lambda}_0}$ and $\gamma'(0) = X$.
\end{definition}

\begin{remark}
The above definition can be shown to be independent of the specific curve chosen \cite{caja12_variedades}.
\end{remark}

\begin{proposition}
\label{prop:bures_metric}
Let $\rho_{\boldsymbol{\lambda}} \in \mathcal{B}$ be a Gibbs state, and let $X, Y \in T_{\rho_{\boldsymbol{\lambda}}} \mathcal{B}$ be tangent vectors. The Bures-Wasserstein metric is given by
\begin{equation}
    g_{\mathrm{BW}}(X, Y) = \mathrm{Re} \left[ \tr \left( L_X \rho_{\boldsymbol{\lambda}} L_Y \right) \right],
\end{equation}
where $L_X$ is the symmetric logarithmic derivative satisfying $\rho_{\boldsymbol{\lambda}} L_X + L_X \rho_{\boldsymbol{\lambda}} = 2 X(\rho_{\boldsymbol{\lambda}})$, and similarly for $L_Y$. This defines a Riemannian metric on $\mathcal{B}$.
\end{proposition}
\begin{proof}
The proof is detailed in \cite{oostrum22_bures}. In brief, the metric $g_{\mathrm{BW}}$ arises from the infinitesimal form of the Bures-Wasserstein distance, where $L_X$ solves the Lyapunov equation for the perturbation $X(\rho_{\boldsymbol{\lambda}})$. The real part ensures symmetry, and the trace preserves positive-definiteness, establishing $g_{\mathrm{BW}}$ as a Riemannian metric.
\end{proof}
In local coordinates, the metric components are
\begin{equation}
g_{ij}(\boldsymbol{\lambda}) = \textrm{Re} \left[ \tr \left( \rho_{\boldsymbol{\lambda}} L_i L_j \right) \right], \ L_i = \frac{\partial \ln (\rho_{\boldsymbol{\lambda}})}{\partial \lambda_i},
\end{equation}
with line element
\begin{equation}
    ds^2 = \sum_{i,j=1}^n g_{ij} d\lambda^i d\lambda^j.
\end{equation}
\begin{corollary}
\label{cor:riemannian}
$(\mathcal{B}, g_{\mathrm{BW}})$ is a Riemannian manifold.
\end{corollary}
The Riemannian structure of $(\mathcal{B}, g_{\mathrm{BW}})$ is further supported by the following fundamental result.

\begin{theorem}[Hopf-Rinow]
\label{thm:hopf_rinow}
For a connected Riemannian manifold $(M, g)$, the following are equivalent:
\begin{enumerate}
\item $M$ is complete as a metric space.
\item $M$ is geodesically complete.
\item A subset in $M$ is compact if and only if it is closed and bounded.
\end{enumerate}
Additionally, if any of the above holds, then any two points $p, q \in M$ can be joined by a minimizing geodesic.
\end{theorem}

\section{Geodesics and quasistatic processes}

The Bures-Wasserstein metric $g_{\mathrm{BW}}$ endows the submanifold of Gibbs states $\mathcal{B}$ with a Riemannian structure, enabling the geometric framework for analyzing quasistatic thermodynamic processes.

\begin{definition}
Let $(\mathcal{B}, g_{\mathrm{BW}})$ be the Riemannian manifold of Gibbs states.
 Let $\boldsymbol{\lambda}(t) = (\lambda_1(t), \dots, \\\lambda_n(t))$ be a curve in coordinate space. The corresponding path in $\mathcal{B}$ is then $\gamma(t) = \rho_{\boldsymbol{\lambda}(t)}: [0,T] \to \mathcal{B}$. The \emph{thermodynamic length} $L(\gamma)$ is
\begin{equation}
L(\gamma)= \int_0^T \sqrt{g_{\mathrm{BW}}(\gamma'(t), \gamma'(t))}  dt
          = \int_0^T \sqrt{\sum_{i,j=1}^n g_{ij}(\boldsymbol{\lambda}(t)) \lambda_i'(t) \lambda_j'(t)}  dt,
\end{equation}
where $\gamma'(t) = \sum_{i=1}^n \lambda_i'(t) \partial_{
\lambda^i} \big|_{\gamma(t)} \in T_{\gamma(t)}\mathcal{B}$ is the pushforward of the coordinate velocity, and $\lambda_i'(t) = d\lambda_i(t)/dt$, for $i=1,\ldots,n$.
\end{definition}

The Bures-Wasserstein distance quantifies the distinguishability between Gibbs states $\rho_{\boldsymbol{\lambda}}$, encoding quantum fluctuations in response to changes in control parameters. Paths $\gamma(t)$ in $\mathcal{B}$ correspond to sequences of instantaneous equilibrium states.

A quasistatic process varies $\lambda(t)$ sufficiently slowly so that the system remains in the Gibbs state $\rho_{\boldsymbol{\lambda}}$ at each instant. Such processes are reversible, producing zero entropy production. 
Minimizing geodesics on $(\mathcal{B}, g_{\mathrm{BW}})$, which assuming completeness of the Riemannian manifold $(\mathcal{B}, g_{\mathrm{BW}})$ exist by virtue of the Hopf-Rinow theorem, thus represent optimal quasistatic transformations, evolving the system through Gibbs states while minimizing the thermodynamic length $L(\gamma)$. This length quantifies cumulative state change and bounds the minimal work required to drive the transformation.

\begin{definition}
    The \emph{entropy production rate} along a finite-speed path $\gamma: [0,T] \to \mathcal{B}$ is
\begin{equation}  
\varsigma_\gamma (t) = \kappa  g_{\mathrm{BW}}(\gamma'(t), \gamma'(t))
               = \kappa \sum_{i,j=1}^n g_{ij}(\boldsymbol{\lambda}(t)
               )\lambda_i'(t) \lambda_j'(t),
\end{equation}
where $\kappa$ is a system-dependent constant ensuring $\varsigma(t)$ has units of entropy per time. Total entropy production along $\gamma(t)$ is
\begin{equation}
\Sigma (t)= \int_0^T \varsigma(t)  dt =  \int_0^T \kappa g_{\mathrm{BW}}(\gamma'(t), \gamma'(t))  dt.
\end{equation}
\end{definition}

For finite-speed processes, $\varsigma_\gamma (t) > 0$ due to non-zero velocity $\gamma'(t)$, indicating irreversibility \cite{tejero2025entropy}. In the quasistatic limit ($\lambda_i' \to 0$), $\varsigma_\gamma (t) \to 0$ and $\Sigma(t)  \to 0$. Minimizing geodesics minimize both $L(\gamma)$ and $\Sigma$, optimizing the path to reduce the entropy production, thus achieving maximal reversibility.

The term $g_{\mathrm{BW}}(\gamma', \gamma')$ measures the instantaneous rate of state evolution, driving irreversibility via excitations or non-equilibrium effects.
Minimizing geodesics mitigate these by following the smoothest path in state space, thereby allowing the system to remain in equilibrium at each step, analogous to adiabatic transformations.
The proportionality $\varsigma_\gamma \propto g_{\mathrm{BW}}(\gamma', \gamma')$ reflects that faster transformations disrupt equilibrium more severely, increasing $\Sigma$.
The following results formalize this connection between geodesics and quasistatic processes.

\begin{proposition}\label{prop:quasistatic_geodesics}
Let $(\mathcal{B}, g_{\mathrm{BW}})$ be the Riemannian manifold of equilibrium states. Optimal quasistatic quantum thermodynamic processes are minimizing geodesics in $(\mathcal{B}, g_{\mathrm{BW}})$.
\end{proposition}
\begin{proof}
Quasistatic processes evolve through Gibbs states $\rho_{\boldsymbol{\lambda}}$ with minimal dissipation. The Bures-Wasserstein metric quantifies state distinguishability, and minimizing geodesics minimize thermodynamic length $L(\gamma)$, corresponding to paths of least entropy production $\Sigma$, as required for optimal quasistatic transformations.
\end{proof}

\begin{theorem}[Geodesic connectivity]\label{thm:geodesic_connectivity}
Let $(\mathcal{B}, g_{\mathrm{BW}})$ be a complete Riemannian manifold. Then any two equilibrium states $\rho_{\boldsymbol{\lambda}_1}, \rho_{\boldsymbol{\lambda}_2} \in \mathcal{B}$ are connected by a minimizing geodesic, representing an optimal quasistatic transformation.
\end{theorem}
\begin{proof}
This result is a direct consequence of the Hopf-Rinow theorem.
Since $(\mathcal{B}, g_{\mathrm{BW}})$ is a Riemannian manifold, by the Hopf-Rinow theorem, completeness of $(\mathcal{B}, g_{\mathrm{BW}})$ implies that any two points can be joined by a minimizing geodesic.
\end{proof}
The Hopf-Rinow theorem ensures that $\mathcal{B}$ is a globally accessible manifold for thermodynamic transformations, with minimizing geodesics providing the optimal paths for quasistatic processes, minimizing both thermodynamic length and entropy production.

There might be situations where the curve $\gamma$ in $\mathcal{B}$ does not necessarily represent a geodesic curve.
Such curves, deviating from these shortest paths, describe quasistatic processes with redundant or oscillatory changes in $\boldsymbol{\lambda}$, such as an isothermal process with inefficient parameter adjustments, yet still preserving the Gibbs state. These non-optimal paths incur higher thermodynamic costs, reflecting practical constraints or suboptimal control in quantum protocols, like those in thermal machines or quantum control. For example, an isothermal transformation maintaining $\rho_{\beta}(t)$ at all times is quasistatic (always in equilibrium), but only the geodesics can minimize dissipation.

\section{Boundary effects}\label{sec:boundary}
In this section, we analyze the consequences of the boundary effects in the manifold of quantum states $\mathcal{D}$. While the core construction uses the finite-dimensional manifold $\mathcal{B}$ of Gibbs states generated by the chosen observables, the framework extends naturally to the full interior $\mathcal{D}^\circ$. 
Only in this section, we adopt $\mathcal{D}$ instead of the base $\mathcal{B}$ in order to account for the difference between the interior of the manifold and its boundary.

\begin{proposition}
Let $\mathcal{D}^\circ$ be the $(m^2 - 1)$-dimensional smooth manifold of full-rank density operators on a finite-dimensional Hilbert space $\mathcal{H}$ with $\dim \mathcal{H} = m$, embedded as the interior of the compact convex set $\mathcal{D}$ of all density operators. Then:
\begin{enumerate}
    \item For a state $\rho \in \mathcal{D}^\circ $, with eigenvalues $p_i$ with $i = 1, \dots, m$ (not necessarily all distinct), the entropy extends continuously from $\mathcal{D}^\circ$ to the boundary $\partial \mathcal{D}$.

    \item For any smooth curve $\gamma: [0, 1) \to \mathcal{D}^\circ$ such that $\lim_{t \to 1^-} \gamma(t) = \bar{\rho} \in \partial \mathcal{D}$ with $\mathrm{rank}(\bar{\rho}) = k < m$, the entropy satisfies
\begin{equation}
\lim_{t \to 1^-} S(\gamma(t)) = S(\bar{\rho}),
\end{equation}
where  $S(\bar{\rho}) = - \sum_{i=1}^k \bar{p}_i\ln \bar{p}_i$, and $\bar{p}_i$ with $i = 1, \dots, k$ (not necessarily all distinct) are the positive eigenvalues of $\bar{\rho}$, where $k$ is the rank of $\bar{\rho}$. 
In particular, if $\bar{\rho}$ is a pure state $(k = 1)$, then $S(\gamma(t)) \to 0$, and if $\bar{\rho}$ is maximally mixed on its support, i.e., $\bar{p}_i = 1/k$ for $i=1,\ldots,k$, then $S(\gamma(t)) \to \ln k$.
\end{enumerate}
\end{proposition}

\begin{proof}
First, we establish the continuous extension of $S$ to $\partial \mathcal{D}$. For a state $\bar{\rho} \in \mathcal{D}$ with eigenvalues $p_i \geq 0$, some possibly zero, the entropy is defined by extending the function $f(x) = -x \ln x$ to the domain boundary $x = 0$.
Let $\gamma: [0,1) \to \mathcal{D}^\circ$ be a smooth curve with $\lim_{t \to 1^-} \gamma(t) = \bar{\rho} \in \partial \mathcal{D}$ of rank $k < m$.
In a neighborhood of $\bar{\rho}$, $\gamma(t)$ can be diagonalized: $\gamma(t) = U(t) \operatorname{diag}(p_1(t), \ldots, p_m(t)) U^*(t)$ with $p_i(t) > 0$, $\sum_{i=1}^m p_i(t) = 1$, and $U(t)$ is a $m \times m$ unitary matrix for each value of $t$. As $t \to 1^-$, assume without loss of generality that $p_i(t) \to \bar{p}_i > 0$ for $i=1,\ldots,k$, and $p_i(t) \to 0$ for $i=k+1,\ldots,m$, with $\sum_{i=1}^k \bar{p}_i = 1$.
Then
\begin{equation}
    S(\gamma(t)) = -\sum_{i=1}^k p_i(t) \ln p_i(t) - \sum_{i=k+1}^m p_i(t) \ln p_i(t).
\end{equation}
The first sum converges to $S(\bar{\rho})$ by continuity of $f(x)$ at $\bar{p}_i > 0$.
For the second sum, consider $f(p_i(t)) = -p_i(t) \ln p_i(t)$. As $p_i(t) \to 0^+$, $f(p_i(t)) \to 0$. The sum $\sum_{i=k+1}^m f(p_i(t))$ has $m-k$ terms and each term converges to 0, so $\sum_{i=k+1}^m p_i(t) \ln p_i(t) \to 0$. 
Thus,
\begin{equation}
    \lim_{t \to 1^-} S(\gamma(t)) = S(\bar{\rho}).
\end{equation}
The special cases follow immediately: pure states give $S(\bar{\rho}) = 0$; maximal mixing on support gives $S(\bar{\rho}) = \ln k$.
\end{proof}

Note that $\partial \mathcal{D}$ is stratified by rank, where the rank-$k$ stratum is a smooth manifold. The stratification ensures that $\partial \mathcal{D}$ is decomposed into smooth manifolds (strata), each corresponding to a fixed rank $k$. The entropy $S(\rho)$ for $\rho$ in the rank-$k$ stratum depends on the eigenvalue distribution.

\begin{theorem}[Geometric unattainability of the boundary]
Let $(\mathcal{B}, g_{\mathrm{BW}})$ be the Riemannian manifold of Gibbs states. Then no geodesic $\gamma: [0, T] \to \mathcal{B}$ of finite length can reach $\partial \mathcal{D}$.
\end{theorem}
\begin{proof}
Consider a boundary element described by a state $\rho \in \partial\mathcal{D}$ with $\mathrm{rank}(\rho) < m$. 
A geodesic curve $\gamma(t) : [0,1) \longrightarrow \mathcal{B}$ such that $\lim_{t \to 1^-} \gamma(t) = \rho$ has eigenvalues $p_i(t)$ with at least one $p_k(t) \to 0^+$ under such limit. Using the Bures-Wasserstein metric, the line element $ds^2$ is seen to contain a term proportional to $p_k^{-1}$ when evaluated over $\gamma$. Since the resulting integral is divergent, $L(\gamma)$ diverges.
\end{proof}

A thermodynamic process in the quantum thermodynamic state space $\mathcal{M}$ projects to a curve $\gamma$ in $\mathcal{D}^\circ$ via $\Xi$. Reaching boundary states in $\partial \mathcal{D}$ would require infinite thermodynamic length, hence infinite resources.

\begin{corollary}[Third law of quantum thermodynamics]
Let $\gamma : [0,1) \longrightarrow \mathcal{B}$ be a smooth curve such that $\lim_{t \longrightarrow 1^-} \gamma(t) \in \partial \mathcal{D}$ is a maximally mixed state. Then, the entropy function $S: \mathcal{D}^\circ \to \mathbb{R}^+$ satisfies
\begin{equation}
\lim_{t \longrightarrow 1^-} S(\gamma(t)) = \ln k,
\end{equation}
where $k$ is the rank of the target boundary stratum. Moreover, boundary strata are unattainable along any finite-length geodesic in $\mathcal{B}$.
\end{corollary}
\begin{proof}
From the theorem, no finite-length geodesic can reach $\partial \mathcal{D}$. By the proposition, $S(\gamma(t)) \to \ln k$, as $S$ is continuous on $\mathcal{B}$. For pure states, $k=1$, meaning $S \to 0$ in infinite time.
\end{proof}
\begin{corollary}
There exists no finite-length thermodynamic process transforming a full-rank Gibbs state into a pure state.
\end{corollary}
\begin{remark}
The classical third law states that absolute zero temperature is unattainable in finite steps. Here, zero-entropy states (pure states) play an analogous role, but the unattainability arises purely from the Riemannian geometry of the state space, not from any explicit parameter. This geometric third law holds for any control protocol in $\mathcal{M}$ projecting to $\mathcal{D}^\circ$.
\end{remark}

The continuity of entropy to the boundary and the geometric unattainability of rank-deficient states provide a differential-geometric derivation of the third law in quantum thermodynamics. The divergence of geodesic lengths toward the boundary implies an infinite thermodynamic cost to reach low-entropy states.
Consequences include fundamental limits on finite-time thermodynamic processes, where complete purification, i.e., transforming a thermal state into a pure state, is impossible thermodynamically.

\section{Connections, curvature and holonomy in the quantum thermodynamic fiber bundle}
\subsection{Ehresmann connection}
The fiber $F_\sigma = \Xi^{-1}(\sigma)$ contains all thermodynamic configurations for the fixed quantum state $\sigma \in \mathcal{B}$, with only the equilibrium point $p_\sigma = F_\sigma \cap \mathcal{B}$ lying on the Legendrian submanifold $\mathcal{E}$. Geodesics on $\mathcal{E}$ connect equilibrium states across different fibers, while non-equilibrium points in $F_\sigma \setminus \mathcal{E}$ represent transient configurations. Relaxation paths within $F_\sigma$ converge to $p_\sigma$, and their dissipation can be quantified and minimized using a pseudo-Riemannian metric $g_{\mathcal{M}}$ on $\mathcal{M}$ that extends the Bures-Wasserstein metric $g_{\mathrm{BW}}$ from $\mathcal{E}$.

Thermodynamic processes, either quasistatic evolutions along Gibbs states or driven non-equilibrium dynamics, correspond to paths in $\mathcal{M}$ projecting via $\Xi$ to curves in $\mathcal{B}$. To model these transitions geometrically, we introduce an Ehresmann connection on the fiber bundle, decomposing $T\mathcal{M}$ into vertical and horizontal subbundles. This connection enables parallel transport of thermodynamic states along paths in $\mathcal{B}$, with curvature quantifying non-integrability and holonomy inducing geometric irreversibility in cyclic processes. This is analogous to the holonomies existent in gauge theories.

\begin{definition}
The \emph{vertical subbundle} $\mathscr{V} \subset T\mathcal{M}$ is
\begin{equation}
\mathscr{V}_p = \ker \left(d\Xi_p : T_p \mathcal{M} \to T_{\Xi(p)} \mathcal{B}\right) = T_p F_{\Xi(p)},
\end{equation}
the tangent space to the fiber $F_\sigma$ at $p \in F_\sigma$. Since $\dim \mathcal{M} = 2n+1$ and $\dim \mathcal{B} = n$, we have $\dim \mathscr{V}_p = n+1$. Assuming that $\Xi$ is a submersion, $d\Xi_p$ is surjective, and $\mathscr{V} = \ker (d\Xi)$ is a smooth subbundle.
\end{definition}

\begin{definition}\label{def:ehrensmann}
An \emph{Ehresmann connection} is a smooth horizontal subbundle $\mathscr{H} \subset T\mathcal{M}$ such that
\begin{enumerate}
\item $T_p \mathcal{M} = \mathscr{H}_p \oplus \mathscr{V}_p$ for all $p \in \mathcal{M}$,
\item $d\Xi_p|_{\mathscr{H}_p} : \mathscr{H}_p \to T_{\Xi(p)} \mathcal{B}$ is a linear isomorphism.
\end{enumerate}
The dimension of this subbundle is $\dim \mathscr{H}_p = n$.
\end{definition}

\begin{proposition} Consider a pseudo-Riemannian metric $g_\mathcal{M}$ on $\mathcal{M}$ such that $g_{\mathcal{M}} |_{\mathscr{V}_p}$ is non-degenerate. This metric induces an Ehresmann connection by defining
\begin{equation}
\mathscr{H}_p = \{ v \in T_p \mathcal{M} \mid g_{\mathcal{M}}(v, w) = 0, \ \forall w \in \mathscr{V}_p \}. 
\end{equation}
\end{proposition}
\begin{proof}
For any tangent vector $u \in T_p \mathcal{M}$, decompose $u = h + v$, with $h \in \mathscr{H}_p$, $v \in \mathscr{V}_p$. The condition $g_{\mathcal{M}}(h, w) = 0$ for all $w \in \mathscr{V}_p$ determines $v$ via the linear system
\begin{equation}
g_{\mathcal{M}}(u - v, w_\beta) = 0,
\end{equation}
where $\beta = 1, \ldots, n+1$, for a basis $\{w_\beta\}$ of $\mathscr{V}_p$. The matrix $g_{\mathcal{M}}(w_\alpha, w_\beta)$ is invertible since $g_{\mathcal{M}} |_{\mathscr{V}_p}$ is non-degenerate, ensuring a unique $v$.

For the second condition in Definition \ref{def:ehrensmann}, since $\mathscr{V}_p = \ker d\Xi_p$,
\begin{equation}
d\Xi_p \left(T_p \mathcal{M} \right) = d\Xi_p \left(\mathscr{H}_p\right).
\end{equation}
Given $\dim \mathscr{H}_p = \dim T_{\Xi(p)} \mathcal{B} = n$ and $\Xi$ a submersion, $d\Xi_p |_{\mathscr{H}_p}$ is an isomorphism. The smoothness of $\mathscr{H}$ follows from the smoothness of $g_{\mathcal{M}}$ and $\Xi$.
\end{proof}

In order to define such a metric $g_\mathcal{M}$, the desiderata are:
\begin{enumerate}
    \item $g_{\mathcal{M}}|_{T\mathcal{E}} = g_{\mathrm{BW}}$,
    \item compatibility with the contact structure $\eta$,
    \item positive-definiteness on $\ker(\eta)$ and controlled signature in transverse directions.
\end{enumerate}

The construction of a metric with the previous properties will be based on the following considerations. Earlier, it was argued that each fiber $\Xi^{-1}(\sigma)$ is diffeomorphic to $\mathbb{R}^{n+1}$. We can further establish such fibers to be \textit{graphs} of smooth functions over the variables $(S, \boldsymbol{a})$. We now develop the essential idea behind this fact: since the functions $\mu_1, \dots, \mu_n$ satisfy the so-called joint injectivity condition, for a given fiber $\Xi^{-1}(\sigma)$ each point of the plane $(S, \boldsymbol{a})$ corresponds uniquely to a point of the fiber, given by the solution of the system of equations $\mu_1(S,\boldsymbol{a},\boldsymbol{\lambda}) = c_1, \hdots, \mu_n(S,\boldsymbol{a},\boldsymbol{\lambda}) = c_n$. Such correspondence defines a map $\varphi_{\sigma} : \mathbb{R}^{n+1} \longrightarrow \Xi^{-1}(\sigma)$. Moreover, due to the submersion condition, the inverse function theorem guarantees that $\varphi_{\sigma}$ is a diffeomorphism, so that the fiber is realized as the smooth graph
\begin{equation}
	\Xi^{-1}(\sigma) = \left\{ \left(S, \boldsymbol{a}, \pi \circ \varphi_{\sigma}(S, \boldsymbol{a})\right) \;\big|\; S, \boldsymbol{a} \in \mathbb{R}^{n+1} \right\},
\end{equation}
where $\pi$ is the projection $(S, \boldsymbol{a}, \boldsymbol{\lambda}) \longmapsto \boldsymbol{\lambda}$. In particular, we can use the variables $(S, \boldsymbol{a})$ to furnish a coordinate chart of each fiber, allowing to define metrics and other tensor fields on the fiber in a simple and explicit manner.

In what follows, it will be useful to enlarge the previous construction to a full coordinate chart on the quantum thermodynamic fiber bundle. Given an element $p \in \mathcal{M}$, $p$ belongs to a unique fiber. Under the assumption that the map $\boldsymbol{\lambda} \mapsto \rho_{\boldsymbol{\lambda}}$ is injective, we can assign to $p$ the values $\boldsymbol{\lambda}$ of the only point where the fiber intersects the equilibrium submanifold $\mathcal{E}$. Let us call these $\bar{\boldsymbol{\lambda}}$. Then, we use the diffeomorphism $\varphi_{\sigma}$, where $\sigma$ here denotes the corresponding Gibbs state. Specifically, we compute $\varphi_{\sigma}^{-1}(p) \in \mathbb{R}^{n+1}$ and denote the result by $(\bar{S}, \bar{\boldsymbol{a}})$. In this manner, each point $p \in \mathcal{M}$ is mapped to the $2n+1$ real numbers $(\bar{S}, \bar{\boldsymbol{a}}, \bar{\boldsymbol{\lambda}})$, and this mapping is a diffeomorphism by virtue of the preceding remarks. The new coordinates are adapted to the fiber bundle structure in the sense that two points $p, q \in \mathcal{M}$ belong to the same fiber if any only if their $\bar{\boldsymbol{\lambda}}$ coordinates coincide. We are now ready to provide the expression of a suitable metric on $\mathcal{M}$:

\begin{proposition}
    A pseudo-Riemannian metric on $\mathcal{M}$, $g_\mathcal{M}$, satisfying the former desiderata is 
\begin{equation}
    g_{\mathcal{M}} = g_{\bar{S}}  d\bar{S}^2 + \sum_{i=1}^n g_{\bar{a}_i}  d\bar{a}_i^2 + \Xi^* g_{\mathrm{BW}} + \sum_{i=1}^n h_i (d\bar{S} \otimes d\bar{\lambda}_i + d\bar{\lambda}_i \otimes d\bar{S}),
\end{equation}
where $g_{\bar{S}} \in \mathbb{R}$ controls the entropy direction, which may be negative; $g_{\bar{a}_i} > 0$ are positive-definite on expectation value directions, for all $i =1,\ldots,n$; $\Xi^* g_{\mathrm{BW}} = \sum_{i,j} g_{ij}  d\bar{\lambda}_i  d\bar{\lambda}_j$ pulls back the Bures-Wasserstein metric; and finally $h_i$ are cross-terms ensuring contact compatibility, for all $i =1,\ldots,n$.

\end{proposition}
Note that, on $\mathcal{E}$
\begin{equation}
    g_{\mathcal{M}}|_{T\mathcal{E}} = \Xi^* g_{\mathrm{BW}},
\end{equation}
whereas on a fiber $F_\sigma$,
\begin{equation}
    g_{\mathcal{M}}|_{F_\sigma} = g_{\bar{S}}  d\bar{S}^2 + \sum_{i=1}^n g_{\bar{a}_i}  d\bar{a}_i^2.
\end{equation}
For a path $\gamma: [0,T] \to \mathcal{M}$ with $\gamma(t) = (\bar{S}(t), \bar{\boldsymbol{a}}(t), \bar{\boldsymbol{\lambda}}(t))$, the thermodynamic length is

\begin{equation}
    L(\gamma) = \int_0^T \sqrt{|g_{\mathcal{M}}(\gamma'(t), \gamma'(t))|}  dt,
\end{equation}
where
\begin{equation}
    g_{\mathcal{M}}(\gamma', \gamma') = g_{\bar{S}} \bar{S}'^2 + \sum_{i=1}^n g_{\bar{a}_i} (\bar{a}_i')^2 + \sum_{i,j=1}^n g_{ij} \bar{\lambda}_i' \bar{\lambda}_j' + 2 \sum_{i=1}^n h_i \bar{S}' \bar{\lambda}_i'.
\end{equation}
Along a relaxation path $\gamma(t) \in F_\sigma$
\begin{equation}
    L(\gamma) = \int_0^T \sqrt{\left| \displaystyle g_{\bar{S}} \bar{S}'^2 + \sum_{i=1}^n g_{\bar{a}_i} (\bar{a}_i')^2 \right|}  dt.
\end{equation}
This represents the thermodynamic length for a non-equilibrium point $\gamma(0)$ in the fiber out of the equilibrium submanifold, toward $\gamma(T) \in F_\sigma \cap \mathcal{E}$. In such process, the entropy production rate is
\begin{equation}
    \varsigma(t) = \kappa  g_{\mathcal{M}}(\gamma'(t), \gamma'(t)),
\end{equation}
with total
\begin{equation}
    \Sigma = \kappa \int_0^T g_{\mathcal{M}}(\gamma'(t), \gamma'(t))  dt,
\end{equation}
for $\kappa > 0 $.

The Ehresmann connection enables parallel transport of thermodynamic states.
\begin{definition}\label{def:horizontal_lift}
For a smooth path $\gamma: [0,1] \to \mathcal{B}$, with $\gamma(0) = \sigma_1$, $\gamma(1) = \sigma_2$, and $p_0 \in \Xi^{-1}(\sigma_1)$, the \emph{horizontal lift} is a curve $\tilde{\gamma}: [0,1] \to \mathcal{M}$ satisfying
\begin{equation}
\Xi \circ \tilde{\gamma}(t) = \gamma(t), \quad \tilde{\gamma}(0) = p_0, \quad \tilde{\gamma}'(t) \in \mathscr{H}_{\tilde{\gamma}(t)}.
\end{equation}
\end{definition}
The lift is unique due to the second condition in Definition \ref{def:ehrensmann}. Locally, the lift is determined by the horizontal vector fields of the connection.
From the connection, it is straightforward to define the curvature form, which measures the non-integrability of $\mathscr{H}$.

\begin{example}
    Return to the qubit example. The quantum thermodynamic state space $\mathcal{M} \cong \mathbb{R}^3$ has coordinates $(S, a, \lambda)$. We choose the function $\mu(S, a, \lambda) = \lambda$. The fibers $F_{\rho_\lambda} = \Xi^{-1}(\rho_\lambda)$ are the level sets of $\mu$, which are 2-dimensional submanifolds diffeomorphic to $\mathbb{R}^2$ with coordinates $S, a$ at fixed $\lambda$. The equilibrium submanifold $\mathcal{E}$ is the 1-dimensional Legendrian curve parameterized by $\lambda$
\begin{equation}
a(\lambda) = \tanh\lambda, \ S(\lambda) = \log(2\cosh\lambda) - \lambda \tanh\lambda.
\end{equation}
The contact form vanishes on $\mathcal{E}$, since
\begin{equation}
\frac{da}{d\lambda} = \sech^2\lambda, \ \frac{dS}{d\lambda} = \lambda \, \sech^2\lambda,
\end{equation}
so $dS = \lambda da$, as expected.

The Bures--Wasserstein metric on $\mathcal{E}$ is the quantum Fisher information matrix. The symmetric logarithmic derivative is directly $L = \partial_\lambda \log \rho_\lambda$, which satisfies $\rho_\lambda L + L \rho_\lambda = 2 \partial_\lambda \rho_\lambda$, where $L = -\sigma_z + \tanh\lambda \ \mathrm{Id}$. The metric component is then
\begin{equation}
g_{\lambda\lambda}(\lambda) = \tr(\rho_\lambda L^2) = 1 - \tanh^2\lambda = \sech^2\lambda.
\end{equation}
For a quasistatic path $\gamma(t) = (S,a,\lambda(t))$ on $\mathcal{E}$, the thermodynamic length is
\begin{equation}
L(\gamma) = \int_0^T \sqrt{g_{\lambda\lambda}(\lambda(t))  (\lambda'(t))^2}  dt = \int_0^T |\lambda'(t)| \sech(\lambda(t))  dt.
\end{equation}
Geodesics minimize this length and correspond to optimal quasistatic transformations.

For non-equilibrium dynamics, consider a point $p = (S_0, a_0, \lambda_0) \in F_{\rho_{\lambda_0}} \setminus \mathcal{E}$. A relaxation path within the fiber is $\gamma(t) = (S(t), a(t), \lambda_0)$ with, e.g., a linear relaxation dependence
\begin{equation}
S'(t) = -\kappa_S (S(t) - S(\lambda_0)), \quad a'(t) = -\kappa_a (a(t) - a(\lambda_0)),
\end{equation}
$\kappa_S, \kappa_a > 0$, converging exponentially to the equilibrium point as $t \to \infty$. To quantify the dissipation, the metric on $\mathcal{M}$ reads
\begin{equation}
g_{\mathcal{M}} = C^{-1}  dS^2 + \chi^{-1}  da^2 + g_{\lambda\lambda}  d\lambda^2 + \alpha \lambda (dS \otimes d\lambda + d\lambda \otimes dS),
\end{equation}
where $C > 0$ and $\chi > 0$ are capacity- and susceptibility-like functions (e.g., $C = \sech^2\lambda$, $\chi = 1$), and $\alpha \in \mathbb{R}$ controls coupling. On $\mathcal{E}$, $dS = \lambda  da$ ensures $g_{\mathcal{M}}|_{\mathcal{E}} = g_{\lambda\lambda}  d\lambda^2 = g_{\mathrm{BW}}$. Restricted to the fiber ($d\lambda = 0$),
\begin{equation}
g_{\mathcal{M}}|_{F_{\rho_{\lambda_0}}} = C^{-1}  dS^2 + \chi^{-1}  da^2.
\end{equation}
The thermodynamic length of the relaxation path is
\begin{equation}
L(\gamma) = \int_0^\infty \sqrt{ C^{-1} (S')^2 + \chi^{-1} (a')^2 }  dt.
\end{equation}
Assuming constant $C, \chi$ and $\kappa_S = \kappa_a = \kappa$,
\begin{equation}
L(\gamma) = \sqrt{ \frac{(S_0 - S(\lambda_0))^2}{C} + \frac{(a_0 - a(\lambda_0))^2}{\chi} }.
\end{equation}
The instantaneous entropy production rate is $\varsigma(t) = \kappa g_{\mathcal{M}}(\gamma', \gamma')$, and the entropy production reads
\begin{equation}
\Sigma = \kappa \left( \frac{(S_0 - S(\lambda_0))^2}{2C} + \frac{(a_0 - a(\lambda_0))^2}{2\chi} \right) > 0,
\end{equation}
vanishing only at equilibrium.

After relaxation, quasistatic transformations between $\lambda_0$ and $\lambda_1$ follow geodesics on $\mathcal{E}$, with length
\begin{equation}
L = \int_{\lambda_0}^{\lambda_1} \sech u \ du = \arctan(\sinh\lambda_1) - \arctan(\sinh\lambda_0),
\end{equation}
minimal and reversible (zero excess entropy production) by the properties of the Riemannian manifold $(\mathcal{B}, g_{\mathrm{BW}})$.

This simple example captures the full bundle geometry: fibers encode non-equilibrium redundancy, equilibrium is the unique gauge-fixed section, quasistatic processes are horizontal geodesics, and relaxation incurs geometric dissipation bounded by the pseudo-Riemannian structure.
\end{example}

\subsection{Curvature}
\begin{definition}\label{def:curvature_form}
For horizontal vector fields $X, Y \in  \mathfrak{X}(\mathscr{H})$, the \emph{curvature} is
\begin{equation} R(X,Y) = \proj_{\mathscr{V}} ([X, Y]) \in  \mathfrak{X}(\mathscr{V}), \end{equation}
where $[X,Y]$ is the Lie bracket and $\proj_{\mathscr{V}}: T\mathcal{M} \to \mathscr{V}$ is the projection onto the vertical subbundle.
\end{definition}
The curvature is a $\mathscr{V}$-valued two-form, expressed via the covariant derivative of the connection $\nabla$: $R(X,Y) = \nabla_X Y - \nabla_Y X - [X,Y]$. In a local trivialization of the bundle, let $\mathcal{B}$ have coordinates $\{b_k\}_{k=1}^n$, and fibers have coordinates $\{v_\alpha\}_{\alpha=0}^{n}$, so $\mathcal{M}$ has coordinates $(b_k, v_\alpha)$, and $\Xi(b_k, v_\alpha) = \sigma(b_k)$ 
The vertical subbundle is $\mathscr{V} = \mathrm{span}\{\partial_{v_\alpha}\}_{\alpha=0}^{n}$. Horizontal vector fields are
\begin{equation}
\label{horvector}
e_k = \dfrac{\partial}{\partial b_k} - \sum_{\alpha=0}^{n} \Gamma_{b_k}^{v_\alpha} \dfrac{\partial}{\partial  v_\alpha},
\end{equation}
where $\Gamma_{b_k}^{v_\alpha}$ are connection coefficients satisfying $g_{\mathcal{M}}(e_k, \partial_{v_\beta}) = 0$, for all $\beta=0,\ldots,n$. The curvature components are
\begin{equation} R(e_k, e_l) = \proj_{\mathscr{V}} ([e_k, e_l]). \end{equation}
The Lie bracket gives the general expression
\begin{equation} \label{bracket}
\begin{split}
[e_k, e_l] &= \left[ \dfrac{\partial}{\partial b_k} - \sum_{\alpha} \Gamma_{b_k}^{v_\alpha} \dfrac{\partial}{\partial  v_\alpha}, \dfrac{\partial}{\partial b_l} - \sum_{\alpha} \Gamma_{b_l}^{v_\alpha} \dfrac{\partial}{\partial  v_\alpha} \right] \\
&= \sum_{\alpha} \left( \dfrac{\partial}{\partial b_k} \Gamma_{b_l}^{v_\alpha} - \dfrac{\partial}{\partial b_l} \Gamma_{b_k}^{v_\alpha} + \sum_{\beta} [\Gamma_{b_k}^{v_\beta}, \Gamma_{b_l}^{v_\alpha}] \right) \dfrac{\partial}{\partial  v_\alpha}, 
\end{split}
\end{equation}
where the commutator is in the Lie algebra of the fiber structure group. The curvature two-form is
\begin{equation} R = \sum_{\alpha} \sum_{k,l} R_{b_k b_l}^{v_\alpha} db_k \wedge db_l \otimes \dfrac{\partial}{\partial  v_\alpha}. 
\end{equation}
\subsection{Holonomy}
Holonomy arises when parallel transporting along a closed loop in $\mathcal{B}$.
\begin{definition}\label{def:holonomy}
For a closed loop $\gamma: S^1 \to \mathcal{B}$ based at $\sigma \in \mathcal{B}$, with horizontal lift $\tilde{\gamma}: S^1 \to \mathcal{M}$ starting at $p_0 \in \Xi^{-1}(\sigma)$, the \emph{holonomy} is the vertical displacement
\begin{equation} \mathrm{Hol}(\gamma, p_0) = \tilde{\gamma}(1) - p_0 \in \mathscr{V}_{p_0}, \end{equation}
where $\tilde{\gamma}(1) \in \Xi^{-1}(\sigma)$.
\end{definition}
The idea is that for a closed loop $\gamma$ in $\mathcal{B}$, the horizontal lift starting at $p_0 \in \Xi^{-1}(\sigma)$ may end at a different point $p_1 \equiv \tilde{\gamma}(1) \in \Xi^{-1}(\sigma)$, with the displacement $h = p_1 - p_0 \in \mathscr{V}_{p_0}$. The holonomy is an element of the fiber structure group, acting as a translation in the thermodynamic variables.

\begin{proposition}\label{proposition:holonomy}
For a closed loop $\gamma$ bounding a surface $\mathscr{S} \subset \mathcal{B}$, the holonomy is
\begin{equation} \mathrm{Hol}(\gamma, p_0) = \mathcal{P}\exp\left(- \int_\mathscr{S} R \right), \end{equation}
where the integral is the fiber-valued integral of the curvature form over $\mathscr{S}$ for the abelian case.
\end{proposition}
\begin{proof}
The connection one-form $\omega \in \Lambda^1(\mathcal{B})$ satisfies $\nabla_X Y = [X, Y] - \omega(X) Y$ for horizontal vector fields $X, Y.$ The curvature is $R = d\omega + \omega \wedge \omega$, such that for a loop $\gamma$, the holonomy is the path-ordered exponential
\begin{equation}\label{eqn:genera_holonomy}
\mathrm{Hol}(\gamma, p_0) = \mathcal{P}\exp \left( - \oint_\gamma \omega \right).
\end{equation}
where $\mathcal{P}$ denotes path ordering along $\gamma$, and $\omega$ is the $\mathfrak{g}$-valued connection 1-form.
For abelian structure groups (e.g., translations in $\mathbb{R}^{n+1}$), this simplifies to
\begin{equation} \mathrm{Hol}(\gamma, p_0) = \exp \left(- \oint_\gamma \omega \right) =\exp\left( - \int_\mathscr{S} d\omega\right) = \exp \left( - \int_\mathscr{S} R \right),
\end{equation}
by Stokes theorem, assuming $\omega \wedge \omega = 0$. For later interest, in the quantum thermodynamic fiber, the fiber is diffeomorphic to $\mathbb{R}^{n+1}$, with abelian translations, so the result applies.
\end{proof}
\begin{remark}
The path ordering is essential since $[R(t_1), R(t_2)] \neq 0$ in general, for $t_1 \neq t_2$, for all $t_1, t_2 \in S^1$. 
The result depends on the homotopy class of $\gamma$ within $\pi_1(\mathcal{B})$.

Even in the non-abelian case, holonomy represents a geometric phase acquired during a cyclic evolution of the parameters $\bar{\boldsymbol{\lambda}}$. Upon returning to the same quantum state $\sigma$, the thermodynamic labels $(\bar{S}, \bar{\boldsymbol{a}})$ are transformed by a non-trivial group element $h \neq e$, for all $h \in G$, and $e$ is the identity element of $G$. 
This transformation is not removable by local reparameterization and induces irreversible entropy production when the system is forced to return to its initial thermodynamic state. 

We shall see in the following section that the fiber is $\mathbb{R}^{n+1}$ with additive group structure, which is abelian. The non-abelian generalization would arise if the thermodynamic labels were subject to a non-commutative redundancy (e.g., in systems with internal symmetries or constrained ensembles), but such extensions lie beyond the current scope.
\end{remark}

\begin{proposition}\label{prop:flat_connection}
If the connection coefficients vanish in a trivialization, then $R = 0$, and the holonomy of any closed loop is zero.
\end{proposition}
\begin{proof}
If the connection coefficients $\Gamma = 0$, horizontal vector fields are $X = \sum_k X^k \partial_{b_k}$. The Lie bracket is
\begin{equation}[X,Y] = \sum_{k,l} \left(X^k \dfrac{\partial}{\partial  b_k} Y^l - Y^k \dfrac{\partial}{\partial  b_k} X^l\right) \dfrac{\partial}{\partial  b_l},
\end{equation}
which lies in $\mathscr{H}$, so $\proj_{\mathscr{V}}([X,Y]) = 0$, hence $R = 0$. The connection one-form $\omega = 0$, so $d\omega = 0$, and
\begin{equation} \mathrm{Hol}(\gamma, p_0) = \exp \left(- \oint_\gamma \omega \right)= 0.\end{equation}
\end{proof}
Some corollaries of particular relevance for quantum thermodynamics are the following:
\begin{corollary}
Any cyclic thermodynamic process corresponding to a closed loop $\gamma$ in $\mathcal{B}$ with non-zero holonomy $\mathrm{Hol}(\gamma, p_0) \neq 0$ induces geometric irreversibility, contributing to net entropy production $\Sigma \neq 0$, unless the connection is flat, reflecting irreversibility due to curvature.
\end{corollary}

\begin{corollary}
If the connection is flat $(R = 0)$, cyclic processes in $\mathcal{B}$ are geometrically reversible: the horizontal lift of $\gamma$ returns to the initial point in the fiber, with zero geometric entropy production $\Sigma = 0$.
\end{corollary}
This means cyclic processes (loops in parameter space like temperature) may return the quantum state but shift thermodynamic variables (e.g., entropy), leading to irreversibility---requiring dissipation to \emph{correct} the shift.

\begin{remark}
The holonomy $\mathrm{Hol}(\gamma, p_0)$ is a functional not of the specific path $\gamma: S^1 \to \mathcal{B}$, but of its homotopy class $[\gamma] \in \pi_1(\mathcal{B}, \sigma)$. That is, if $\gamma \sim \gamma'$ via a homotopy fixing the base point $\sigma$, then $\mathrm{Hol}(\gamma,p_0) = \mathrm{Hol}(\gamma',p_0)$. 
This follows from the horizontal nature of the lift: any deformation of $\gamma$ within $\mathcal{B}$ induces a corresponding deformation of $\tilde{\gamma}$ within $\mathcal{M}$ that remains horizontal, hence returns to the same fiber point after closure.

In the abelian case, this implies that holonomy defines a group homomorphism $\Phi: \pi_1(\mathcal{B}) \to G$, where $\Phi([\gamma]) = \mathrm{Hol}(\gamma, p_0)$, independent of base point $p_0$, since vertical translations commute. The image of $\Phi$ is generated by curvature integrals over a basis of 2-cycles in $H_2(\mathcal{B})$, via de Rham cohomology.

In the non-abelian case, holonomy defines a representation of the fundamental group $\rho^G: \pi_1(\mathcal{B}) \to G$, where conjugacy classes of $\rho^G([\gamma])$ are base-point independent, and the full holonomy group is the subgroup generated by such elements under concatenation of loops. Non-commutativity implies that the order of traversal matters: $\rho^G([\gamma_1 \cdot \gamma_2]) \neq \rho^G([\gamma_1]) \rho^G([\gamma_2])$ in general. 

Two cyclic protocols in intensive parameter space $\bar{\boldsymbol{\lambda}}(t)$ that return to the same Gibbs state $\sigma$ but traverse homotopically distinct paths accumulate different geometric entropy shifts, even in the quasistatic limit. This constitutes a topological contribution to irreversibility---unremovable by slowing down the process---and sets a fundamental lower bound on dissipation in cyclic processes operating over topologically non-trivial control manifolds. The \emph{topological irreversibility} is then tied to the concept of \emph{topological entropy production}. 

In the current framework, where $\mathcal{B}$ is contractible in many finite-dimensional models (e.g., full-rank Gibbs states), $\pi_1(\mathcal{B}) = 0$, and holonomy vanishes for contractible loops. However, in constrained ensembles, $\pi_1(\mathcal{B})$ may be non-trivial, opening the door to \emph{topological quantum thermodynamics}. 
These topological considerations are further expanded in Sec. \ref{subsec:topological} by elevating the bundle to a principal bundle structure.
\end{remark}

The vertical subbundle corresponds to shifts in thermodynamic labels at fixed quantum state, modeling internal adjustments without altering the density operator. On the other hand, the horizontal subbundle enables consistent evolution of labels as the state changes, akin to parallel transport in gauge theories.
The pseudo-Riemannian metric $ g_{\mathcal{M}}$ extends the Bures-Wasserstein metric off equilibrium, with a potentially indefinite signature along the entropy direction suggesting a \emph{timelike} character reminiscent of spacetime metrics. 
Curvature measures the non-integrability of horizontal fields, sourcing holonomy as a thermodynamic geometric phase in cyclic processes, which shifts labels (e.g., inducing an extra $\Delta \bar{S}$), and necessitates dissipative resets to close the thermodynamic cycle. Flat connections permit fully reversible cycles without geometric losses, whereas non-zero curvature introduces irreducible irreversibility even in quasistatic limits. The contact non-integrability and positive-definite aspects on $\ker(\eta)$ underpin the second law, ensuring some sort of non-zero entropy production. 
Therefore, the unification of thermodynamics with gauge-like principles offers analogs to effects like the Aharonov-Bohm phase for electrodynamics.
From a practical standpoint, the direct consequences involve geometric dissipation in cyclic processes in the quantum regime, such as quantum heat engines.

\section{Principal Bundle Structure}
\subsection{Principal bundle}
To elevate the quantum thermodynamic fiber bundle $(\mathcal{M}, \mathcal{B}, \Xi, F)$ to a \emph{principal bundle}, we recall the definition of a principal $G$-bundle. 
\begin{definition}
    A \emph{principal $G$-bundle} consists of a \emph{total space} $P$, a \emph{base manifold} $B$, a smooth surjective \emph{projection} $\pi: P \to B$, and a Lie group $G$ acting smoothly, freely, and transitively on the right on each fiber $\pi^{-1}(b) \cong G$, such that $B = P/G$ is the \emph{orbit space}. The bundle admits local trivializations $\pi^{-1}(U) \cong U \times G$ for an open cover $\{U_i \}$ of $B$, with transition functions taking values in $G$.
\end{definition}
\begin{theorem}[Principal $\mathbb{R}^{n+1}$-bundle structure]
The quantum thermodynamic fiber bundle $(\mathcal{M}, \mathcal{B},\\ \Xi, F)$ admits a principal $\mathbb{R}^{n+1}$-bundle structure.
\end{theorem}
\begin{proof}
Identify the fiber $F \cong \mathbb{R}^{n+1}$ with the additive Lie group $(\mathbb{R}^{n+1}, +)$. Define the right action of $G = \mathbb{R}^{n+1}$ on the total space $\mathcal{M}$ by
\begin{equation}
\label{eq:group_action}
(\bar{S}, \bar{\boldsymbol{a}}, \bar{\boldsymbol{\lambda}}) \cdot g = (\bar{S} + \delta \bar{S}, \bar{\boldsymbol{a}} + \delta \bar{\boldsymbol{a}}, \bar{\boldsymbol{\lambda}}),
\end{equation}
for $g = (\delta \bar{S}, \delta \bar{\boldsymbol{a}}) \in \mathbb{R}^{n+1}$. This action is \emph{free}, since $p \cdot g = p$ implies $g = (0, \boldsymbol{0})$, the identity. It is \emph{transitive} on each fiber $F_\sigma = \Xi^{-1}(\sigma)$, because for any two points $p_1 = (\bar{S}_1, \bar{\boldsymbol{a}}_1, \bar{\boldsymbol{\lambda}})$, $p_2 = (\bar{S}_2, \bar{\boldsymbol{a}}_2, \bar{\boldsymbol{\lambda}}) \in F_\sigma$, the group element $g = (\bar{S}_2 - \bar{S}_1, \bar{\boldsymbol{a}}_2 - \bar{\boldsymbol{a}}_1)$ satisfies $p_1 \cdot g = p_2$. The action preserves the fibers (fixed $\bar{\boldsymbol{\lambda}}$).

The base manifold $\mathcal{B}$ is the quotient $\mathcal{M} / \mathbb{R}^{n+1}$, with each orbit \footnote{The quotient map
$\mathcal{M} \rightarrow \mathcal{M}/\mathbb{R}^{n+1} =\mathcal{B}, \ p\;\mapsto\;[p]=\operatorname{Orb}(p)$
sends every point to its \emph{orbit}, defined as the set of all points in $\mathcal{M}$ that are reached from a given starting point by acting with the group $\mathbb{R}^{n+1}$. Since the group acts freely and transitively on each fiber, the orbit of any point $p\in\mathcal{M}$ is exactly the whole fiber that contains $p$. (i.e., each equivalence class) corresponding to a fiber $F_\sigma$. Local triviality follows from the global coordinate structure $\mathcal{M} \cong \mathbb{R}^{n} \times \mathbb{R}^{n+1}$, where the first factor parametrizes $\mathcal{B}$ and the second corresponds to the group $G$.} Thus, $(\mathcal{M}, \mathcal{B}, \Xi, \mathbb{R}^{n+1})$ is a smooth principal $\mathbb{R}^{n+1}$-bundle with abelian structure group.
\end{proof}

This principal structure interprets the redundancy in thermodynamic labeling as a \emph{gauge symmetry}: group actions correspond to translations in $(\bar{S}, \bar{\boldsymbol{a}})$ at fixed $\bar{\boldsymbol{\lambda}}$, preserving the physical quantum state $\sigma$ but altering its thermodynamic description---analogous to gauge freedom in field theories, where different potentials describe the same physical field.

\begin{definition}
A \emph{principal connection} (or \emph{connection $1$-form}, or \emph{gauge field}) on the principal $\mathbb{R}^{n+1}$-bundle $(\mathcal{M}, \mathcal{B}, \Xi, \mathbb{R}^{n+1})$ is a smooth $\mathfrak{g}$-valued $1$-form $
\omega \in \Lambda^1(\mathcal{M}; \mathfrak{g})$,
where $\mathfrak{g} = \mathbb{R}^{n+1}$ is the Lie algebra of the structure group $G = (\mathbb{R}^{n+1}, +)$, satisfying the following two properties:
\begin{enumerate}
    \item For all $g \in G$, $R_g^* \omega = \mathrm{Ad}_{g^{-1}} \circ \omega$ where $R_g: \mathcal{M} \to \mathcal{M}$, $p \mapsto p \cdot g$, is the right action of the group, and $\mathrm{Ad}: G \to \mathrm{Aut}(\mathfrak{g})$ is the adjoint representation.  
    Since $G$ is abelian, $\mathrm{Ad}_g Y = Y$ for all $g, Y \in \mathfrak{g}$, so the condition simplifies to
    \begin{equation}
        R_g^* \omega = \omega, \ \text{for all } g \in \mathbb{R}^{n+1}.
    \end{equation}
    This means the connection is \emph{invariant under global translations} in the fiber coordinates $(\bar{S}, \bar{\boldsymbol{a}})$.

    \item For all $X \in \mathfrak{g}$, let $\xi_X \in \mathfrak{X}(\mathcal{M})$ be the \emph{fundamental vector field} defined by
    \begin{equation}
    \xi_X(p) = \left. \frac{d}{dt} \right|_{t=0} \bigl( p \cdot \exp(tX) \bigr).
    \end{equation}
    Since $\exp(tX) = tX$ for $G = \mathbb{R}^{n+1}$, in coordinates we have $\xi_X = X^0 \partial_{\bar{S}} + \sum_{i=1}^n X^i \partial_{\bar{a}_i} $.
    The connection must satisfy
    \begin{equation}
        \omega(\xi_X) = X,\ \text{for all }  X \in \mathfrak{g}.
    \end{equation}
    This identifies vertical tangent vectors with infinitesimal group translations.
\end{enumerate}
\end{definition}

The principal $\mathbb{R}^{n+1}$-bundle structure $(\mathcal{M}, \mathcal{B}, \Xi, \mathbb{R}^{n+1})$ reveals that the redundancy in thermodynamic labeling: multiple values of $(\bar{S}, \bar{\boldsymbol{a}})$ corresponding to the same physical quantum state $\sigma \in \mathcal{B}$ at fixed intensive parameters $\bar{\boldsymbol{\lambda}}$ is a \emph{gauge symmetry}.

\begin{definition}
Let $\Xi: \mathcal{M} \to \mathcal{B}$ be the principal quantum thermodynamic bundle. A \emph{global gauge transformation} is a bundle automorphism $f: \mathcal{M} \to \mathcal{M}$,
i.e., a diffeomorphism satisfying:
\begin{enumerate}
    \item $\Xi \circ f = \Xi$,
    \item $f(p \cdot g) = f(p) \cdot g$, for all $p \in \mathcal{M}$, $g \in G$.
\end{enumerate}
The set of all such $f$ forms the gauge group $
\mathcal{G}(\mathcal{M}) = \mathrm{Aut}(\mathcal{M})$. A \emph{local gauge transformation} is a bundle automorphism on $\Xi^{-1}(U) \to U$, where $U \subset \mathcal{B}$.
\end{definition}

In local coordinates $( \bar{S}, \bar{\boldsymbol{a}},\bar{\boldsymbol{\lambda}})$ on $\mathcal{M}$, a gauge transformation $f \in \mathcal{G}(\mathcal{M})$ acts as
\begin{equation}
f(\bar{S}, \bar{\boldsymbol{a}}, \bar{\boldsymbol{\lambda}}) = \bigl( \bar{S} + \phi_{\bar{S}}(\bar{\boldsymbol{\lambda}}), \bar{\boldsymbol{a}} + \boldsymbol{\phi}_{\bar{\boldsymbol{a}}}(\bar{\boldsymbol{\lambda}}), \bar{\boldsymbol{\lambda}} \bigr),
\end{equation}
where $\phi_{\bar{S}}: \mathcal{B} \to \mathbb{R}$ and $\boldsymbol{\phi}_{\bar{\boldsymbol{a}}}: \mathcal{B} \to \mathbb{R}^n$ are smooth functions. This translates thermodynamic labels within each fiber while preserving the physical state: $\Xi(f(p)) = \Xi(p) = \sigma$; and the intensive parameters: $\bar{\boldsymbol{\lambda}}(f(p)) = \bar{\boldsymbol{\lambda}}(p)$.

Assuming the injectivity of the Gibbs map $  \boldsymbol{\lambda}\mapsto\rho_{\boldsymbol{\lambda}}  $ (see Remark~\ref{remark:gibbs_injectivity_zeroth}), the equilibrium submanifold $  \mathcal{E}  $ intersects each fiber $  F_\sigma  $ at exactly one point. Equilibrium states therefore provide a natural gauge-fixing condition: for every physical quantum state $  \sigma\in\mathcal{B}  $ there exists a unique thermodynamic representative on $  \mathcal{E}  $, thereby eliminating the gauge redundancy carried by the fibers.

\begin{remark}
Note that the \emph{horizontal distribution} can be defined as the kernel of $\omega$:
\begin{equation}
    \mathscr{H}_p = \ker \omega_p = \{ v \in T_p\mathcal{M} \mid \omega(v) = 0 \}.
\end{equation}
This complements the vertical subbundle $\mathscr{V}_p = T_p F_{\Xi(p)} \cong \mathfrak{g}$ and defines parallel transport.    
\end{remark}

\subsection{Curvature on the principal bundle}
In a local trivialization $\Xi^{-1}(U) \cong U \times \mathbb{R}^{n+1}$ of the principal $\mathbb{R}^{n+1}$-bundle $(\mathcal{M}, \mathcal{B}, \Xi, \mathbb{R}^{n+1})$, let $U \subset \mathcal{B}$ be an open set with coordinates $\{ \bar{\lambda}_k \}_{k=1}^n$ (so $b_k = \bar{\lambda}_k$), and let the fiber $\mathbb{R}^{n+1}$ have coordinates $\{ v_\alpha \}_{\alpha=0}^n$, where $v_0 = \bar{S},$ and $v_i = \bar{a}_i, \ \text{for } i=1,\dots,n$. 
The left-invariant Maurer--Cartan form $\theta$ on the abelian Lie group $G = (\mathbb{R}^{n+1}, +)$ is
\begin{equation}
\theta = \sum_{\alpha=0}^n dv_\alpha \otimes \dfrac{\partial}{\partial  v_\alpha}.
\end{equation}
The principal connection $1$-form $\omega \in \Lambda^1(\mathcal{M}; \mathfrak{g})$ is defined globally, but it will be useful to consider the local form it takes in our particular trivialization. Concretely, knowing that the horizontal vector fields given by Eq. (\ref{horvector}) must satisfy $\omega(e_k)=0$ for all $k = 1, \dots, n$, the connection form can be shown to be given by
\begin{equation}
\label{eq:local_omega_standard}
\omega = \sum_{\alpha=0}^n \sum_{k=1}^n \left( dv_\alpha - \Gamma^k_\alpha  d\bar{\lambda}_k \right) \otimes \dfrac{\partial}{\partial  v_\alpha}.
\end{equation}
This is the pullback of the Maurer-Cartan form minus the gauge potential $\Gamma^k_\alpha$. The coefficients $\Gamma^k_\alpha(\bar{\boldsymbol{\lambda}})$ encode the geometric coupling between changes in intensive parameters $\bar{\boldsymbol{\lambda}}$ and shifts in thermodynamic labels $(\bar{S}, \bar{\boldsymbol{a}})$ during parallel transport. The connection coefficients $\Gamma^k_\alpha$ are determined by the orthogonality condition
$g_{\mathcal{M}}(e_k, \partial_{v_\beta}) = 0$. Explicitly:

\begin{itemize}
\item For $\beta = 0$, i.e., $\partial_{v_0} = \partial_{\bar{S}}$:
\begin{equation}
g_{\mathcal{M}}\left(e_k, \dfrac{\partial}{\partial  \bar{S}}\right) = g_{\mathcal{M}}\left(\dfrac{\partial}{\partial  \bar{\lambda}_k} - \sum_{\alpha=0}^n \Gamma^k_\alpha \dfrac{\partial}{\partial  v_\alpha},  \dfrac{\partial}{\partial  \bar{S}} \right) = g_{\bar{S} \bar{\lambda}_k} - \Gamma^k_0  g_{\bar{S}\bar{S}} = h_k - \Gamma^k_0  g_{\bar{S}} = 0.
\end{equation}
Thus, $\Gamma^k_0 = h_k / g_{\bar{S}}$.

\item For $\beta = i \in \{1,\dots,n\}$, i.e., $\partial_{v_i} = \partial_{\bar{a}_i}$:
\begin{equation}
\begin{split}
g_{\mathcal{M}}\left(e_k,  \dfrac{\partial}{\partial  \bar{a}_i}\right) &= g_{\mathcal{M}}\left(\dfrac{\partial}{\partial  \bar{\lambda}_k} - \sum_{\alpha=0}^n \Gamma^k_\alpha \dfrac{\partial}{\partial  v_\alpha},  \dfrac{\partial}{\partial  \bar{a_i}} \right) = -\Gamma^k_i  g_{\bar{a}_i \bar{a}_i} = -\Gamma^k_i  g_{\bar{a}_i} = 0.
\end{split}
\end{equation}
Thus, $\Gamma^k_i = 0, \text{ for all } i=1,\dots,n$.
\end{itemize}

Therefore, the connection coefficients are
\begin{equation}
\label{eq:Gamma_explicit}
\Gamma^k_\alpha(\boldsymbol{\bar{\lambda}}) =
\begin{cases}
\displaystyle h_k(\bar{\boldsymbol{\lambda}}) / g_{\bar{S}}(\bar{\boldsymbol{\lambda}}) & \text{if } \alpha = 0, \\
0 & \text{if } \alpha = 1,\dots,n.
\end{cases}
\end{equation}
\begin{remark}
Note that the connection acts exclusively in the entropic direction, reflecting that changes in intensive parameters $\bar{\boldsymbol{\lambda}}$ induce adjustments in the value of $\bar{S}$ but not in the value of the coordinates $\bar{a}_i$, consistent with the structure of the quantum thermodynamic fiber bundle.    
\end{remark}

\begin{proposition}
Since $\mathfrak{g}$ is abelian, the curvature $2$-form is $R = d\omega$. In the local trivialization, for horizontal lifts
\begin{equation}
    e_k = \dfrac{\partial}{\partial \bar{\lambda}_k} - \sum_{\alpha=0}^n \Gamma^k_\alpha  \dfrac{\partial}{\partial  v_\alpha},
\end{equation}
for all $k=1, \ldots, n$, the curvature evaluated on basis vectors is
\begin{equation}
\label{eq:curvature_general}
R(e_k, e_l) = \sum_{\alpha=0}^n \left( \dfrac{\partial}{\partial \bar{\lambda}_k} \Gamma^l_\alpha -  \dfrac{\partial}{\partial \bar{\lambda}_l} \Gamma^k_\alpha \right) \dfrac{\partial}{\partial  v_\alpha}.
\end{equation}
\end{proposition}

\begin{proof}
The local connection form is
\begin{equation}
\omega = \sum_{\alpha=0}^n \sum_{m=1}^n \left( dv_\alpha - \Gamma^m_\alpha  d\bar{\lambda}_m \right) \otimes \dfrac{\partial}{\partial  v_\alpha}.
\end{equation}
Its exterior derivative is directly given by
\begin{equation}
\begin{split}
d\omega &= \sum_{\alpha=0}^n \sum_{k=1}^n \sum_{l=1}^n \left(\dfrac{\partial}{\partial \bar{\lambda_k}} \Gamma^k_\alpha \right) d\bar{\lambda}_l \wedge d\bar{\lambda}_k \otimes \dfrac{\partial}{\partial  v_\alpha} \\
&= \sum_{\alpha=0}^n \sum_{k=1}^n \sum_{l=1}^n \frac{1}{2} \left(  \dfrac{\partial}{\partial \bar{\lambda}_k} \Gamma^l_\alpha -  \dfrac{\partial}{\partial \bar{\lambda}_l} \Gamma^k_\alpha \right) d\bar{\lambda}_k \wedge d\bar{\lambda}_l \otimes \dfrac{\partial}{\partial  v_\alpha}.
\end{split}
\end{equation}
Evaluating on the horizontal lifts $e_k, e_l$, the curvature is
\begin{equation}
R(e_k, e_l) = d\omega(e_k, e_l) = \sum_{\alpha=0}^n \left( \dfrac{\partial}{\partial \bar{\lambda}_k} \Gamma^l_\alpha -  \dfrac{\partial}{\partial \bar{\lambda}_l} \Gamma^k_\alpha \right) \dfrac{\partial}{\partial  v_\alpha},
\end{equation}
which is Eq.~\eqref{eq:curvature_general}. 
Note that this expression is similar to the general one given by Eq. (\ref{bracket}). The Lie bracket vanishes since the group is abelian.
\end{proof}
Substituting the explicit connection coefficients from Eq.~\eqref{eq:Gamma_explicit}, we have that only the $\alpha = 0$ term contributes to the curvature
\begin{align}
R(e_k, e_l) &= \left[\dfrac{\partial}{\partial \bar{\lambda}_k} \left( \frac{h_l}{g_{\bar{S}}} \right) - \dfrac{\partial}{\partial \bar{\lambda}_l} \left( \frac{h_k}{g_{\bar{S}}} \right) \right] \dfrac{\partial}{\partial \bar{S}}.
\end{align}
Thus, $R$ is a $\mathfrak{g}$-valued 2-form acting solely in the \emph{entropic} direction, reflecting that non-commutativities in intensive parameter variations induce non-integrable shifts in thermodynamic entropy.

\begin{remark}
The curvature $R$ of the principal connection is a $\mathfrak{g}$-valued $2$-form on $\mathcal{M}$, where $\mathfrak{g} = \mathbb{R}^{n+1}$ is the Lie algebra of the structure group $G = (\mathbb{R}^{n+1}, +)$. In the local trivialization, it takes the form
\begin{equation}\label{eqn:curvarure_form}
R = \frac{1}{2} \sum_{k,l=1}^n \left[ \frac{\partial}{\partial \bar{\lambda}_k} \left( \frac{h_l}{g_{\bar{S}}} \right) - \frac{\partial}{\partial \bar{\lambda}_l} \left( \frac{h_k}{g_{\bar{S}}} \right) \right] d\bar{\lambda}_k \wedge d\bar{\lambda}_l \otimes \frac{\partial}{\partial \bar{S}}.
\end{equation}
Since the structure group is abelian, $R = d\omega$. The tensor product with $\partial_{\bar{S}}$ reflects that the curvature takes values exclusively in the entropic component of $\mathfrak{g}$; all components in the $\partial_{\bar{a}_i}$ directions vanish due to the absence of cross terms coupling $d\bar{a}_i$ and $d\bar{\lambda}_j$ in the metric $g_{\mathcal{M}}$. The expression given by Eq. (\ref{eqn:curvarure_form}) is the standard coordinate expression of a Lie-algebra-valued $2$-form (see e.g.~\cite{nakahara2003geometry,Bleecker1981Gauge}).
\end{remark}

\begin{corollary}
Curvature acts exclusively in the entropic direction, $\partial_{\bar{S}}$, measuring the non-integrability of the entropy compensation induced by changes in the intensive parameters $\bar{\boldsymbol{\lambda}}$.
\end{corollary}

\begin{remark}
The connection is \emph{flat} if and only if
\begin{equation}
    \dfrac{\partial}{\partial \bar{\lambda}_k} \left( \frac{h_l}{g_{\bar{S}}} \right) = \dfrac{\partial}{\partial \bar{\lambda}_l} \left( \frac{h_k}{g_{\bar{S}}} \right),
\end{equation}
for all $k,l$, i.e., the $1$-form 
\begin{equation}
    \psi_i = \dfrac{h_i}{g_{\bar{S}}} d\bar{\lambda}_i,
\end{equation}
is closed for all $i=1,\ldots,n$.
\end{remark}

\begin{corollary}
By Proposition~\ref{proposition:holonomy}, the holonomy of a closed loop $\gamma: S^1 \to \mathcal{B}$ bounding a surface $\mathscr{S} \subset \mathcal{B}$ is
\begin{equation}
\mathrm{Hol}(\gamma, p_0) = -\int_{\mathscr{S}} R,
\end{equation}
a vertical displacement in the fiber. Non-zero holonomy implies that parallel transport along a cycle in the space of Gibbs states shifts the thermodynamic coordinates $(\bar{S}, \bar{\boldsymbol{a}})$, inducing \emph{geometric irreversibility} $\Sigma \neq 0$ unless $R = 0$.
\end{corollary}

In this gauge framework, quasistatic transformations correspond to paths with minimal holonomic deviation. Irreversibility in cyclic processes emerges as a curvature effect, unifying thermodynamic dissipation with geometric phases in gauge theories. This mirrors the Aharonov-Bohm phase in electromagnetism or Berry's phase in quantum mechanics, where holonomy in a principal bundle over parameter space leads to observable shifts despite local flatness. Here, non-zero $R$ generates non-integrable shifts in $(\bar{S}, \bar{\boldsymbol{a}})$, requiring dissipative corrections to close the cycle, manifesting as $\Sigma \neq 0$---a gauge-induced irreversibility intrinsic to the thermodynamic bundle geometry.

If the principal connection on the quantum thermodynamic bundle is flat, then all cyclic processes in the base manifold $\mathcal{B}$, corresponding to closed loops in intensive parameters $\bar{\boldsymbol{\lambda}}$, induce zero holonomy. 
This implies that parallel transport along such cycles returns the system to its initial thermodynamic configuration without any geometric shift in labels $(\bar{S}, \bar{\boldsymbol{a}})$, implying no additional entropy production from bundle geometry. 
In practical terms, flat connections characterize systems whose thermodynamic cycles have no gauge-like losses, provided the parameter space $\mathcal{B}$ possesses trivial topology.

The principal fiber bundle structure introduces gauge symmetry into quantum thermodynamics. The additive group acts by translating the coordinates $ (\bar{S}, \bar{\boldsymbol{a}})$ while preserving the quantum state and intensive parameters, representing a redundancy in thermodynamic descriptions akin to gauge freedoms in field theories where different potentials give the same physical fields.
Equilibrium serves as a gauge-fixing condition, ensuring a unique canonical representation per state due to the maximum entropy principle. 
The connection 1-form $  \omega  $ acts as a gauge field, coupling changes in parameters to label shifts, with (in this case) its abelian curvature $  R = d\omega$ (or $d\omega + \omega\wedge \omega$ in the non-abelian case) focused on the entropy direction, measuring non-commutativity in parameter evolutions.
As mentioned in the previous section, the direct consequences include holonomy-imposed bounds on quantum engine efficiencies, where cyclic protocols accumulate unremovable dissipation. 
This opens the avenue for optimization of control protocols exploiting gauge choices, and conceptual parallels to thermodynamic analogs of the Aharonov-Bohm effect, potentially inspiring experimental probes of geometric phases in thermal machines.

\begin{remark}[Second law of quantum thermodynamics]
\label{remark:quantum_second_law}
The quantum second law emerges from the non-integrability of the contact structure and the geometry of non-equilibrium paths. For a finite-speed process along a path $\gamma(t)$ in $\mathcal{E}$, the entropy production rate is $\varsigma(t) = \kappa g_{\mathrm{BW}}(\gamma'(t), \gamma'(t)) \geq 0$, with total $\Sigma = \kappa \int_0^T g_{\mathrm{BW}}(\gamma'(t), \gamma'(t))  dt \geq 0$. Irreversibility is minimal along geodesics and vanishes in the quasistatic limit. Off equilibrium, in fibers $F_\sigma$ and assuming that the full bundle metric $g_{\mathcal{M}}$ is Riemannian, relaxation paths incur positive entropy production $\Sigma > 0$ due to deviations from $\mathcal{E}$. Curvature-induced holonomy in cyclic processes generates additional geometric entropy production.
\end{remark}

\subsection{Gauge and physical interpretation of curvature}
The curvature of the principal connection on the quantum thermodynamic fiber bundle is the standard curvature $2$-form of mathematical gauge theory, as defined in standard references such as \cite{hamilton2017mathematical, nakahara2003geometry, Bleecker1981Gauge}. This object admits several complementary physical interpretations that clarify its thermodynamic significance.

First, in the Ehresmann picture, the curvature measures the failure of horizontal integrability. Specifically, $R(X,Y) = \proj_{\mathscr{V}}([X,Y])$ for horizontal vector fields $X,Y \in \mathfrak{X}(\mathscr{H})$ quantifies how much parallel transport around an infinitesimal closed loop fails to return to the initial point in the fiber due to the non-integrability of the horizontal distribution $\mathscr{H}$. In local coordinates, $R(e_k, e_l)$ is the vertical projection of the Lie bracket of the horizontal lifts $e_k, e_l$. Non-zero curvature then implies that thermodynamic labels $(\bar{S}, \bar{\boldsymbol{a}})$ do not return to their initial values after a closed cycle in intensive-parameter space $\bar{\boldsymbol{\lambda}}$. This is the primary geometric mechanism underlying holonomy-induced irreversibility in the framework.

In the language of gauge field theory, the curvature $R \equiv F = d\omega + \omega \wedge \omega$ is equivalent to the strength of the local gauge field. In a local trivialization, $R$ acts as the force field felt by the thermodynamic coordinates when transported along the base manifold. In the abelian case, $R$ takes values exclusively in the entropy direction $\partial_{\bar{S}}$, inducing a net shift $\Delta \bar{S}$ during cyclic variations of $\bar{\boldsymbol{\lambda}}$---analogous to how the electromagnetic field strength $F_{\mu\nu}$ generates forces via the covariant derivative $D = d + A$ whose square is $D^2 = F$.

From a topological viewpoint, $R$ is the source of holonomy. Parallel transport around a closed loop $\gamma$ in the base generates the holonomy $\mathrm{Hol}(\gamma, p_0) = \mathcal{P}\exp(-\oint_\gamma \omega)$, which, in the abelian case and by Stokes' theorem, reduces to $\exp(-\int_{\mathscr{S}} R)$ for any surface $\mathscr{S}$ bounding $\gamma$. This holonomic shift in the thermodynamic coordinates, typically entropy, requires dissipative corrections to restore equilibrium, providing a purely geometric origin for irreversibility in cyclic processes. This interpretation is the direct thermodynamic counterpart of the Berry phase in quantum mechanics and the Aharonov--Bohm effect in gauge theory.
By the Ambrose--Singer theorem, the curvature generates the holonomy group: the Lie algebra of the holonomy group is spanned by the values of $R$ evaluated on horizontal vectors. Non-zero curvature thus produces a non-trivial holonomy group acting on the fiber, manifesting as topological irreversibility in cyclic thermodynamic protocols.

Finally, $R$ is simply a Lie algebra-valued $2$-form $R \in \Lambda^2(\mathcal{M}; \mathfrak{g})$ with $\mathfrak{g} = \mathbb{R}^{n+1}$, encoding the (abelian) non-commutativity of infinitesimal gauge transformations. In potential non-abelian extensions (e.g., systems with internal symmetries or constrained ensembles), the $\omega \wedge \omega$ term would appear, leading to non-commuting holonomies consistent with Wilczek--Zee phases, see Section~\ref{subsec:wilczek}.

\subsection{Toward \emph{topological quantum thermodynamics}: invariants and topological considerations in the quantum thermodynamic bundle}\label{subsec:topological}

Following the exposition in Ref. \cite{hamilton2017mathematical}, in the principal bundle $(\mathcal{M}, \mathcal{B}, \Xi, \mathbb{R}^{n+1})$, quantities unchanged under bundle automorphisms, i.e., arbitrary re-gaugings of the thermodynamic labels $(\bar{S},\bar{\boldsymbol{a}})$ while keeping the physical state $\sigma\in\mathcal{B}$ fixed, are \emph{gauge invariants}.
These invariants encode the intrinsic geometry and topology of the connection $  \omega  $ and its curvature $  R  $, furnishing universal, gauge-independent lower bounds on irreversible entropy production in cyclic quantum thermodynamic processes.

Differences of connections $\omega_1 - \omega_2$ transform as 1-forms on $\mathcal{B}$, forming an affine space of connections. 
The curvature descends to a well-defined closed 2-form $R_{\mathcal{B}}\in\Lambda^2(\mathcal{B};\mathbb{R}^{n+1})$ on the base, quantifying label shifts independent of the gauging.
\emph{Bianchi's identity} $dR + [\omega, R] = 0$, where $dR=0$ for the abelian case, ensures closure of $R_{\mathcal{B}}$. That is, since the structure group is abelian, Bianchi's identity reduces to $  dR_{\mathcal{B}}=0  $, implying that the holonomy around any closed loop $  \gamma\subset\mathcal{B}  $ depends only on its homotopy class $  [\gamma]\in\pi_1(\mathcal{B})  $. Consequently, topologically equivalent cycles induce identical geometric entropy shifts.

A natural gauge-invariant diagnostic for a loop $  \gamma:S^1\to\mathcal{B}  $ is the Wilson loop on an associated vector bundle $E=\mathcal{M}\times_{\rho}V$
\begin{equation}
W^E_\gamma(\omega) = \mathrm{tr} \left[ \mathcal{P} \exp \left( -\oint_\gamma \rho_* \omega^s \right) \right],
\label{eq:wilson_loop}
\end{equation}
where $\omega^s = s^* \omega$ is the pull-back of the connection in a local trivialization over an open set containing $\gamma$ with global gauge $s: U \to \mathcal{M}$; and $V$ is the internal vector space on which the connection acts through the representation $\rho: \mathbb{R}^{n+1} \to \mathrm{GL}(V)$, the trivial representation on the fiber of $E$.

In the abelian case, path-ordering reduces to the ordinary exponential and the trace is unnecessary for line bundles, corresponding to scalar holonomy shifts in labels, so
\begin{equation}
W^E_\gamma(\omega) = \exp \left( -\oint_\gamma \omega^s \right).
\end{equation}
When projected onto the entropy component of the fiber, this gives a real number $\Delta\bar{S}_{\mathrm{geo}}$, measuring the net shift in thermodynamic labels after parallel transport. This geometric contribution adds to the total entropy production $\Sigma$ in cyclic processes, providing a strict lower bound on dissipation for any cyclic protocol. 

Topologically, with an $\textrm{Ad}$-invariant scalar product $\langle \cdot, \cdot \rangle$ on $\mathfrak{g} = \mathbb{R}^{n+1}$, e.g., the Euclidean scalar product, the \emph{Chern--Simons 3-form} on $\mathcal{M}$ is defined as
\begin{equation}
    \alpha(\omega) = \langle \omega, R \rangle \in \Lambda^3(\mathcal{M}),
\end{equation}
satisfying $d\alpha = \langle R, R \rangle$. 
Here, $\langle R, R \rangle$ acts as a local geometric measure of dissipation induced by curvature. Under automorphisms, $\alpha$ transforms by an exact term plus a closed 3-form from the structure group, so the integral
\begin{equation}
   \mathcal{S}_{\mathrm{CS}}(\omega) = \int_{\mathscr{M}} \alpha(\omega) \pmod{\mathbb{Z}},
\end{equation}
over closed oriented 3-submanifolds $\mathscr{M} \subset \mathcal{M}$, e.g., a region swept by a continuous family of loops in parameter space, is gauge-invariant and metric-independent. This number is known as \emph{Chern--Simons action}.
A non-vanishing value highlights the existence of \emph{topologically protected} fractional contributions to entropy production: quasistatic cycles enclosing non-contractible 3-volumes cannot be made dissipation-free, no matter how slowly they are driven.
Note that this Chern--Simons invariant is purely topological, and does not depend on the metric chosen.
 
Flat connections give $\mathcal{S}_{\mathrm{CS}}=0$, permitting reversible cycles. Differences $\omega - \omega_0$ transform adjointly, and $\mathcal{S}_{\mathrm{CS}}$ classifies topological sectors---universal bounds on dissipation in quantum engines over non-simply connected parameter spaces, e.g., periodic $\boldsymbol{\lambda}$ resulting in toroidal $\mathcal{B}$.

In natural extensions where the abelian group is compactified, e.g., $\mathbb{R}\to S^1$ arising from periodic intensive parameters or modular identifications, additional topological invariants emerge from characteristic classes.
For a principal $S^1$-bundle $  \mathcal{M}\to\mathcal{B}$ equipped with the connection $\omega$, the curvature defines the first Chern class
\begin{equation}
    c_1(\mathcal{M}) = -\frac{1}{2\pi i} [R_{\mathcal{B}}] \in H^2_{\mathrm{dR}}(\mathcal{B}; \mathbb{R}),
\end{equation}
a real de Rham cohomology class independent of the choice of connection. 
A non-trivial $c_1$ precludes the existence of a global section, that is, the existence of a unique global thermodynamic labeling. It also forces irreducible holonomy around non-contractible loops, even if the curvature vanishes locally.

\begin{example}
In particular, consider the canonical Hopf bundle $S^1 \hookrightarrow S^3 \xrightarrow{\pi} S^2 \cong \mathbb{CP}^1$ \cite{hatcher2002algebraic}. The total space is defined as the unit sphere in $\mathbb{C}^2$, $S^3 = \{ (z_1, z_2) \in \mathbb{C}^2 | |z_1|^2 + |z_2|^2 = 1 \}$. The projection map $\pi: S^3 \to \mathbb{CP}^1$ sends each point to the complex line it spans.
The fiber over each point of the base $S^2$ is a circle $S^1$ consisting of all global phases $e^{i\theta}(z_1,z_2)$ consistent with a given projective point in $\mathbb{CP}^1$. By taking the curvature 2-form $R$ corresponding to a Dirac monopole field of unit strength, we can generate the fundamental class of $H^2(S^2;\mathbb{Z})\cong\mathbb{Z}$. Therefore, the first Chern class is represented by
\begin{equation}
  c_1(\mathrm{Hopf}) = \left[ \frac{R}{2\pi} \right],
\end{equation}
and integrating it over the base space gives
\begin{equation}
    \int_{S^2}c_1(\mathrm{Hopf})=1.
\end{equation}
Note that $\pi_1(S^2)=0$: every closed loop in $S^2$ is contractible. From the long exact sequence of homotopy groups of the fibration, the connecting homomorphism yields an isomorphism $\pi_2(S^2) \cong \pi_1(S^1) \cong \mathbb{Z}$. This shows that the non-triviality of the bundle is inherently tied to the second homotopy group wrapping around the fiber, which is equivalently measured by the first Chern class.
In our thermodynamic interpretation, this integer $1$ is the minimal topological charge: any cyclic driving protocol whose control manifold contains a 2-cycle homologous to the generator of $H_2(\mathcal{B})$ must induce at least one unit of topological entropy shift $\Delta\bar{S}_{\mathrm{geo}}=1$.
\end{example}

For non-abelian compactifications, e.g., SU(2) bundles for symmetry-protected degeneracies in spin systems, higher Chern classes $c_k(P)\in H^{2k}_{\mathrm{dR}}(\mathcal{B})$ impose further constraints on geometric dissipation through integrated curvature. 

Therefore, all these invariants unify topological dissipation. Wilson loops bound local cycle production, while Chern--Simons integrals and Chern classes set universal lower limits on entropy production in topologically non-trivial control protocols. Geometric irreversibility is thereby elevated to a gauge- and topology-protected phenomenon.

\subsection{Classical limit}
Classical thermodynamics is recovered if and only if the fibers $F_\sigma$ of the principal bundle $(\mathcal{M}, \mathcal{B}, \Xi, \mathbb{R}^{n+1})$ are singletons, i.e., $\dim F_\sigma = 0$ for all $\sigma \in \mathcal{B}$. This trivializes the bundle and makes the group action the identity. Consequently, $\mathcal{M}$ becomes isomorphic to the classical contact manifold $M$.
This condition is both necessary and sufficient, as it eliminates quantum redundancies: in quantum systems, the fibers parameterize multiple labels per state, whereas classically, each state has a unique label, causing the fibers to collapse to points. $\Xi$ becomes bijective (the identity on reduced space), aligning the corresponding structures: $\eta$ matches classical Gibbs 1-form, $\mathcal{E}$ reduces to Legendrian $E$, and $g_{\mathrm{BW}}$ becomes classical Fisher-Rao.

Mathematically, the bundle trivializes to $\mathcal{M} \cong M \times \{e\}$, where $\{e\}$ is the trivial group. This preserves tangent bundles, contact distributions, and Reeb fields.
This criterion unifies the limits, confirming that the quantum theory is the general case and the classical theory emerges as its limit.

\begin{example}
For the qubit case, in the high-temperature limit, where $\lambda \equiv \beta\to 0$, any full-rank qubit state approaches the maximally mixed $\rho = \mathrm{Id}/2$, with $a \to 0$ and $S \to \ln 2$. The fiber $F_{\mathrm{Id}/2}$ collapses effectively toward the equilibrium point $( \ln 2, 0, 0 )$, as small deviations in $\mu$ produce states close to uniform. Physically, at infinite temperature, all states become indistinguishable from the unique maximum-entropy state, and label redundancy vanishes, mirroring classical thermodynamics where extensive variables uniquely determine the state without quantum ambiguities.

Therefore, non-thermal deviations become negligible, fibers trivialize to points, and the bundle reduces to the classical contact manifold with unique thermodynamic coordinates per state, recovering the classical framework.
\end{example}

\section{Infinite-Dimensional Extensions}
\label{sec:infinite}

The framework developed herein relies on a finite-dimensional Hilbert space $\mathcal{H} \cong \mathbb{C}^m$ with $m < \infty$, ensuring that $\mathscr{B}(\mathcal{H}) \cong M_m(\mathbb{C})$ is compact and the state space $\mathcal{D}$ is finite-dimensional. However, most relevant quantum systems are modeled using infinite-dimensional separable Hilbert spaces, including $L^2(\mathbb{R}^d)$ for quantum fields or the bosonic Fock space $\mathscr{F}(\mathbb{C})$ for harmonic oscillators. It may thus appear that we are led to the consideration of infinite-dimensional manifolds, such as Banach manifolds. It is indeed the case that the set of positive density operators can be endowed with a Banach manifold structure. However, as soon as we choose a \textit{finite} collection of observables $\{ A_1, \dots, A_n \}$, the ensuing geometric structures will remain finite-dimensional. In other words, it will not be necessary to stray from finite-dimensional differential geometry to fulfill our purposes. The manipulation of linear maps defined over infinite-dimensional domains does introduce some subtleties leading to potential alterations to our original formalism.

Let us first focus on the simplest scenario in which we choose to work with a single observable $A$. The associated Gibbs states will be proportional to the exponential $\exp(-\lambda A)$. Recall that for self-adjoint operators with possibly continuous spectrum, the previous exponential is defined by means of the spectral theorem. Here comes the fundamental discrepancy between the finite and infinite-dimensional cases: whereas the trace of the exponential of a matrix is always finite, it may fail to converge for an operator living on an infinite-dimensional Hilbert space. This will generically be the case if $A$ has a continuous spectrum. However, if the spectrum is discrete, convergence of the trace relies on how fast the eigenvalues increase. Thus, it will usually happen that the trace of $\exp(-\lambda A)$ converges only for $\lambda > \lambda_c$. Gibbs states can only be defined as long as this condition is met. For this reason, the domain of the map $\Xi$, which is later used to construct the fiber bundle, has to be restricted accordingly.

Assuming that the set of admissible values of $\lambda$ is non-empty (if it were empty, one has to resort to other notions of thermal states, such as the so-called Kubo--Martin--Schwinger (KMS) states \cite{kubo57,MS59}), there exist two qualitatively different possible forms for this set: it is either $(\lambda_c, +\infty)$ or $[\lambda_c, +\infty)$. The first case simply entails that the quantum thermodynamic phase space is $\mathbb{R}^2 \times (\lambda_c, +\infty)$ instead of $\mathbb{R}^3$. Since the two are diffeomorphic, our overall fiber bundle construction remains essentially unaltered, as long as the $\mu$ function is properly chosen. The following provides an explicit example:

\begin{example}
    Let us consider the quantum harmonic oscillator, where the Hilbert space is the bosonic Fock space, as mentioned earlier. Choose as the only observable the number operator $N = a^{\dagger}a$. Gibbs states will be given by
    \begin{equation}
        \rho_{\lambda} = \frac{\exp(-\lambda N)}{\mathrm{tr}[\exp(-\lambda N)]}
    \end{equation}
    as long as the denominator is finite. It can be seen that
    \begin{equation}
    \mathrm{tr}[\exp(-\lambda N)] = \frac{1}{1-e^{-\lambda}}
    \end{equation}
    when $\lambda > 0$. Otherwise, the series is divergent. Therefore, the map carrying each set of thermodynamic labels to a thermal state is to be defined as $\Xi: \mathcal{M} \to \mathcal{B}$ where $\mathcal{M} = \mathbb{R}^2 \times (0,+\infty)$ and
\begin{equation}
\Xi(S,a,\lambda)
 = \frac{\exp\left(-\mu(S,a,\lambda) N\right)}
       {\tr\left[\exp\left(-\mu(S,a,\lambda) N\right) \right]},
\end{equation}
where, apart from the general conditions, the codomain of $\mu$ must now be $(0,+\infty)$. Next, we will determine the equilibrium submanifold, which is given by a curve inside the phase space $\mathcal{M}$. We have
\begin{equation}
    a(\lambda) = \langle N \rangle_{\rho_{\lambda}} = \frac{1}{e^{\lambda}-1}.
\end{equation}
\begin{equation}
    S(\lambda) = \mathrm{tr}(\rho_{\lambda} \log \rho_{\lambda}) = \frac{\lambda}{e^{\lambda}-1} - \log(1-e^{-\lambda}).
\end{equation}
Note that both $a(\lambda)$ and $S(\lambda)$ are well-defined for all the admissible values of $\lambda$. Therefore, the equilibrium submanifold $\mathcal{E}$ corresponds to the curve $(S(\lambda), a(\lambda), \lambda)$ for $\lambda \in (0,+\infty)$. 

As a non-trivial example of $\mu$ function, we may choose
\begin{equation}
    \mu(S,a,\lambda) = a^2\lambda (e^{\lambda}-1)^2.
\end{equation}
This function has the following properties:
\begin{itemize}
    \item When restricted to the equilibrium submanifold
    \begin{equation}
        \mu(S(\lambda),a(\lambda),\lambda) = \lambda.
    \end{equation}
    \item $\mu$ is well-defined on $\mathbb{R}^2 \times (0,+\infty)$. Moreover, $\mu > 0$ over its domain.
    \item For fixed $S$ and $a$, the equation $\mu(S,a,\lambda) = c$, $c > 0$ has a unique positive solution in $\lambda$.
\end{itemize}
Since all the required conditions are satisfied, the previous choice of $\Xi$ gives rise to a fiber bundle with base space diffeomorphic to the equilibrium submanifold and fibers provided by the level sets of $\mu$.

    \end{example}

\vspace{2pc}

As pointed out before, the set of admissible values of $\lambda$ being an open interval $(\lambda_c, +\infty)$ is not the only possibility. Depending on the choice of quantum system and observable, the critical value $\lambda_c$ may be included in the interval. As a consequence, the geometric structure would change in a non-trivial manner: the space $\mathbb{R}^2 \times [\lambda_c, +\infty)$ can be understood as a \textit{manifold with boundary}, the boundary being in this case the set $\mathbb{R}^2 \times \{ \lambda_c \}$. The appearance of such boundaries can serve as a crude model of a \textit{phase transition}. In this case, the two phases are distinguished in terms of the existence or non-existence of thermal states for the system.

In the more general situation, we consider a finite set of observables $\{ A_1,  \dots, A_n \}$. Provided that the exponential $\exp(-\sum_{i=1}^n \lambda_n A_n)$ is well-defined, the condition that it be trace-class will impose a limitation on the permissible values of the intensive parameters $(\lambda_1, \dots, \lambda_n) \in \mathbb{R}^n$. The ensuing fiber bundle construction is to be restricted to this subset. Crucially, said subset may contain some part of its boundary, in analogy to the one-dimensional situation. As earlier pointed out, the emergence of boundaries can provide a description of phase transitions at the quantum level. The extent to which the formalism presented here ought to be modified to allow for the proper modelization of phase transitions, criticality, and other related phenomena lies beyond the scope of the present work.

\begin{remark}
  In cases where the explicit Gibbs construction is ill-defined, thermal equilibrium is instead characterized by KMS states, as noted earlier. While every KMS state reduces to the ordinary Gibbs state in finite dimensions, the KMS condition in the general case provides a complete and intrinsic encoding of thermal equilibrium without relying on a well-defined density matrix.
This implies that the map assigning thermodynamic labels to equilibrium states can be replaced by the KMS map $\boldsymbol{\lambda} \mapsto \omega_{\boldsymbol{\lambda}}$ defined over the admissible domain of intensive parameters. Under this substitution, the base manifold $\mathcal{B}$ of the quantum thermodynamic fiber bundle becomes the manifold of KMS states parametrized by $\boldsymbol{\lambda}$. The natural Riemannian metric on this space of states is the Bogoliubov--Kubo--Mori metric \cite{petz1993bogoliubov}. With this metric, the geometric framework may be extended naturally to thermal quantum field theory.
\end{remark}

\section{Quantum thermodynamic protocols and non-abelian structures}
\label{sec:applications}
To illustrate the practical utility of our geometric framework, we apply it to several quantum thermodynamic systems and protocols.
For simplicity of physical interpretation, in the following subsections, we take $\mu_i(S, \boldsymbol{a}, \boldsymbol{\lambda}) = \lambda_i$. It then follows that the new thermodynamic coordinates $(\bar{S}, \bar{\boldsymbol{a}}, \bar{\lambda})$ reduce to the original labels $(S, \boldsymbol{a}, \boldsymbol{\lambda})$.

\subsection{Cyclic processes and the thermodynamic Berry phase}
Quantum heat engines operate cyclically, converting heat into work using quantum working media, such as qubits or quantum harmonic oscillators \cite{bhattacharjeeQuantumThermalMachines2021c,tejero2024,tejero2026}. In our geometric framework, a cyclic process corresponds to a closed loop $\gamma: S^1 \to \mathcal{B}$ in the base manifold of Gibbs states. The horizontal lift $\tilde{\gamma}$ to the total space $\mathcal{M}$ via the principal connection describes the evolution of thermodynamic labels $(S, \boldsymbol{a})$.

For a non-flat connection, parallel transport along $\gamma$ leads to a non-trivial holonomy. This holonomy manifests as a vertical displacement in the fiber, shifting the coordinates upon cycle completion, even in the quasistatic limit. To close the thermodynamic cycle (return to the initial labels), this shift must be corrected via dissipative relaxation within the fiber, incurring entropy production $\Sigma > 0$.

In the abelian case (structure group $\mathbb{R}^{n+1}, +$), the holonomy reduces to a scalar phase factor. As an example, for a two-dimensional parameter space $\boldsymbol{\lambda} = (\lambda_1, \lambda_2)$, the holonomy is given by the line integral of the connection or, equivalently, by Stokes' theorem as the surface integral of the curvature
\begin{equation}\label{eqn:deltaS_twoparam}
\Delta S = \int_{\mathscr{S}} \left[ \frac{\partial}{\partial \lambda_1} \left( \frac{h_{\lambda_2}}{g_S} \right) - \frac{\partial}{\partial \lambda_2} \left( \frac{h_{\lambda_1}}{g_S} \right) \right] d\lambda_1 \, d\lambda_2,
\end{equation}
where $\mathscr{S}$ is any surface bounded by $\gamma$.

This geometric irreversibility is analogous to the Berry phase in quantum mechanics \cite{berry1984quantal}, where adiabatic cyclic parameter evolution induces a phase shift in the wavefunction. Here, we term it the \emph{thermodynamic Berry phase}, or \emph{holonomic thermodynamic phase}, as it arises from the gauge structure of the principal bundle over thermodynamic parameter space. 
Such phase has been observed and identified in models of quantum heat engines \cite{giri2017geometric, bhandari2020geometric}. The geometric construct provides an elegant and fundamental explanation of their appearance.

\begin{example}
Consider a two-level quantum Otto engine with Hamiltonian $A \equiv H = \omega \sigma_z / 2$, coupled alternately to hot ($\beta_h$) and cold ($\beta_c$) reservoirs. The cycle is a loop in $\mathcal{B}$ parameterized by $\boldsymbol{\lambda} = (\beta, \omega)$. The Bures-Wasserstein metric on $\mathcal{B}$ results in geodesics for optimal quasistatic strokes, minimizing dissipation per stroke. For finite-speed cycles, holonomy adds a geometric term to the total work $W$ and efficiency of the cycle, $\eta$. In general, we can write the extracted work from the cycle and the efficiency as
\begin{equation}
 W = W_{\mathrm{dyn}} + W_{\mathrm{geo}}, \quad \eta = 1 - \frac{Q_c + Q_{\mathrm{geo}}}{Q_h}, 
\end{equation}
where $W_{\mathrm{geo}}$ and $\Delta Q_{\mathrm{geo}}$ arise from correcting the holonomic shift. 
Explicitly, for a connection with curvature $R \propto d\beta \wedge d\omega$, and in the adiabatic limit, the thermodynamic Berry phase modifies the extracted work by $ W_\mathrm{geo} = T \Delta S_{\mathrm{geo}}$, where $\Delta S_{\mathrm{geo}}$ is the entropy shift, leading to efficiency reductions scaling as the enclosed curvature flux.

For a rectangular cycle in the surface $\mathscr{S}\subset \mathcal{B}$ delimited by $\beta\in [\beta_c, \beta_h]$, and $\omega\in [\omega_1,\omega_2]$, the holonomy can be computed following Eq. (\ref{eqn:deltaS_twoparam}) as
\begin{equation}\label{eqn:deltaS_example}
 \Delta S_{\mathrm{geo}} = \int_{\mathscr{S}} \left[ \dfrac{\partial}{\partial \beta} \left( \dfrac{h_\omega} {g_S} \right) - \dfrac{\partial}{\partial \omega} \left(\dfrac{h_\beta}{g_S} \right) \right] d\beta d\omega.
\end{equation}
This provides a quantifiable geometric bound, e.g., for $\Delta\beta = \beta_h - \beta_c$ and $\Delta\omega = \omega_2 - \omega_1$, $\Delta S_{\mathrm{geo}}$ scales as $\Delta S_{\mathrm{geo}} \sim \overline{R}(\Delta\beta \Delta\omega)$, where $\overline{R}$ is the average curvature on $\mathscr{S}$. In the case of a flat domain, $\overline{R} = 0$ and $\Delta S_{\mathrm{geo}} = 0$, as expected.
\end{example}

\subsection{Non-abelian extensions and the thermodynamic Wilczek–Zee phase}\label{subsec:wilczek}
The principal bundle constructed above has an abelian structure group $(\mathbb{R}^{n+1}, +)$. Natural extensions to non-abelian structure groups appear whenever the redundancy in thermodynamic labeling involves non-commutative transformations. This occurs, for instance, when equilibrium or steady states exhibit degeneracies protected by non-abelian symmetries, or when the effective state manifold requires a non-trivial group action beyond translations.

In such cases, the fiber over a given physical state $\sigma \in \mathcal{B}$ is a non-abelian Lie group $G$, and the principal connection $\omega$ takes values in the Lie algebra $\mathfrak{g}$ of $G$. The curvature two-form then acquires the non-abelian term $\omega \wedge\omega \neq 0$ term. The holonomy is now the path-ordered exponential as in Eq. (\ref{eqn:genera_holonomy}), which generally does not commute for different loops.
The resulting non-commuting shifts of thermodynamic labels provide a geometric description of topological entropy production and are directly linked to fractional statistics in anyonic systems \cite{wen2017zoo}.

The non-abelian holonomy, originally introduced by Wilczek and Zee \cite{wilczek1984appearance}, generalizes Berry's geometric phase \cite{berry1984quantal} to degenerate eigenspaces. When a quantum system has a degenerate energy level, and external parameters are varied adiabatically around a closed path in parameter space, the time-evolution operator restricted to the degenerate subspace acquires a unitary matrix-valued geometric factor---the \emph{Wilczek--Zee holonomy}. This holonomy is the parallel transport operator in a vector bundle over parameter space, with the connection (gauge potential) given by the non-abelian generalization of the Berry connection. 

For a degenerate subspace of dimension $d > 1$, the non-abelian Berry connection is the matrix-valued 1-form whose entries encode the geometric phase acquired when the degenerate states are parallel-transported in parameter space. The curvature is the matrix-valued 2-form $R = dA + i [A, A],$ and the holonomy along a closed loop $\gamma$ is the path-ordered exponential
\begin{equation}\label{eqn:WZ}
\textrm{Hol}_{\mathrm{WZ}}(\gamma) = \mathcal{P} \exp\left( -i \oint_\gamma \sum_k A_kd\lambda_k \right) \in \mathrm{U}(d).
\end{equation}
In the bundle language, this is precisely the non-abelian principal holonomy $\mathrm{Hol}(\gamma, p_0)$, acting by non-commutative transformations on the thermodynamic labels $(S, \boldsymbol{a})$ within the degenerate fiber.

This general, non-trivial holonomy is further interpreted here from a thermodynamic perspective with the topological entropy production, allowing the full thermodynamic description of such \emph{exotic} systems. This non-commuting geometric phase leads to topological entropy production that depends on the homotopy class of the loop in parameter space. Different cyclic protocols returning to the same $\sigma$ but traversing homotopically inequivalent paths accumulate distinct (and non-commuting) entropy shifts, imposing a fundamental, topology-protected lower bound on dissipation.

\subsection{Physical examples of non-abelian systems in condensed matter}
\subsubsection{Bose--Einstein condensates}
The first concrete example arises in spin-1 spinor Bose gases, realized experimentally in spinor Bose--Einstein condensates with traps \cite{ho1998spinor,kawaguchi2012spinor}. The Hamiltonian describing such systems preserves SO(3) rotational symmetry for vanishing magnetic field. 
As a particular example, in the spin-1 ferromagnetic phase, the order-parameter manifold has symmetry $\textrm{SO(3)} \cong \mathbb{RP}^3 $, effectively providing a non-trivial topological structure (see Table 11 in \cite{kawaguchi2012spinor} for more details).

External magnetic fields break the symmetry and serve as intensive parameters in $\boldsymbol{\lambda}$. Cyclic variation of these fields induces an effective non-abelian SO(3) (or SU(2)) gauge connection on the degenerate manifold, leading to the aforementioned Wilczek--Zee holonomy. Thermodynamically, near-zero-temperature Gibbs states are dominated by this degenerate ground manifold; parallel transport over cycles in field space produces non-commuting rotations of the order parameter. 
This holonomy manifests geometrically in thermodynamic observables: it contributes to curvature-dependent entropy shifts, modified susceptibilities, or topological magnetization currents, analogous to geometric responses in spinor condensates \cite{ho1998spinor,kawaguchi2012spinor}.

\subsubsection{Fibonacci anyons}
The most striking realization of non-commutative holonomy in quantum systems is given by \emph{anyons}, quasiparticles with non-trivial phase upon particle exchange \cite{Wilczek_82,Wilczek_06,dunlop2025}.
Consider particularly the case of Fibonacci anyons \cite{nayak2008}, whose fusion rule is $\tau \times \tau = 1 + \tau$. The meaning of this rule is that when two $\tau$ anyons are brought together, they can annihilate into the vacuum (1) with one possible fusion channel, or they can fuse into another $\tau$ anyon with a second possible channel.

Now take the minimal non-trivial degenerate configuration: three $\tau$ anyons whose total topological charge is the vacuum. To reach the vacuum, the first two $\tau$ anyons can fuse in two ways: either they fuse directly to vacuum, and the third $\tau$ must then fuse with that vacuum; or they fuse to a $\tau$, and the third $\tau$ fuses with it to the vacuum.

These two different fusion paths define two orthogonal states that span a two-dimensional fusion space. This is the smallest degenerate subspace that exhibits genuinely non-abelian statistics.
A braiding operation, i.e., moving one $\tau$ anyon around another along a closed path, corresponds to a unitary matrix acting on this two-dimensional fusion space. 
In our bundle language, such a braiding process is exactly parallel transport along a closed loop $\gamma$ in the control manifold $\mathcal{B}$. The holonomy associated with a single braid generator is the path-ordered exponential given by Eq. (\ref{eqn:WZ}), with an element $X_{\mathrm{WZ}}(\gamma)\in \mathrm{U}(2)$.

After the braiding, the system returns to the same physical quantum state, but the thermodynamic labels have been transformed: the initial point $p_0$ in the fiber is moved to a new point $p_1 = X_{\mathrm{WZ}}\cdot p_0$ inside the same fiber. To close the loop, i.e., to properly perform the braiding operation, the length of this shortest path is precisely the geodesic distance in the fiber metric between $p_0$ and the shifted point
\begin{equation}
\Delta S_{\mathrm{geo}} = d_{g_{\mathcal{M}}}(p_0,p_1) = d_{g_{\mathcal{M}}}(p_0,X_{\mathrm{WZ}} \cdot p_0).
\end{equation}
The distance function $d_{g_{\mathcal{M}}}$ is induced by the fiber metric.
\begin{remark}
The group acts by isometries on the fiber, so one could equivalently minimize over equivalent representatives, but for a fixed braid, the element $X_{\mathrm{WZ}}$ is given, and the distance is simply that between $p_0$ and $p_1$.
\end{remark}

Because the holonomy depends only on the homotopy class of the braid, that is, the representation of the braid group, $\Delta S_{\mathrm{geo}}$ is a topological invariant.
Even the simplest non-trivial braid of Fibonacci anyons provides a finite positive value of $\Delta S_{\mathrm{geo}}$. This is the thermodynamic manifestation of topological protection: the same braiding operations necessarily generate an irreducible entropy production.

Experimental platforms now realize Fibonacci anyons in superconducting processors, where braiding operations have been performed and the resulting unitary evolution measured directly \cite{xu2024non}. In these systems, the predicted geometric entropy shift can be observed as excess heat after a full braiding cycle, providing a direct test of the topology-protected dissipation bound derived from the principal-bundle curvature.

These extensions provide geometric and topological tools for analyzing irreversibility, topological order, and dissipation in open and driven systems.
It is worth noting that the experimental realization of non-abelian gauge fields and the direct observation of Wilczek--Zee phases, including Wilson loops, have been achieved, see e.g., Ref. \cite{sugawa2021wilson}.

\section{Conclusion}
In this work, we have presented a geometrical framework for quantum thermodynamics, grounded in contact geometry and principal fiber bundle theory.

The quantum thermodynamic state space is constructed as a contact manifold, with equilibrium Gibbs states forming Legendrian submanifolds that generalize classical thermodynamic geometry to the quantum regime. The principal fiber bundle structure over the manifold of density operators provides a precise description of non-equilibrium configurations: its fibers characterize relaxation processes, while their unique intersections with the equilibrium submanifold ensure thermodynamic consistency.

The Bures-Wasserstein metric on the equilibrium submanifold provides a Riemannian perspective on quasistatic transformations, minimizing dissipation through geodesics, while diverging geodesic length toward rank-deficient states offers a geometric derivation of the third law, underscoring the unattainability of pure states in finite processes. Extensions to non-equilibrium via pseudo-Riemannian metrics and principal connections reveal curvature-induced holonomy as a source of irreversibility in cyclic processes, quantifying entropy production geometrically.

Quantum thermodynamics is geometry.
In the formalism here proposed, the laws emerge as consequences of the manifold, the contact form, and the fiber bundle structure.
From the uniqueness of equilibrium to the unattainability of pure states; from reversible geodesics to irreversible holonomy, the full structure of thermal behavior in quantum systems is encoded in the differential and fiber geometry of the quantum thermodynamic state space.

\section{Data availability}
No datasets were generated or analyzed during the current study.

\section{Acknowledgements}
The authors wish to express their gratitude to Prof. Frederic Schuller, whose lectures inspired them to pursue studies in theoretical and mathematical physics during their undergraduate years.

We acknowledge funding from the Ministry for Digital Transformation and of Civil Service of the Spanish Government through projects PID2021-126217NB-I00, PID2023-149365NB-I00, PID2024-162155OB-I00, FPU20/02835 and QUANTUM ENIA project call - Quantum Spain project, and by the European Union through the Recovery, Transformation and Resilience Plan - NextGenerationEU within the framework of the Digital Spain 2026 Agenda. Finally, we are also grateful for the technical support provided by PROTEUS, the supercomputing center of the Institute Carlos I for Theoretical and Computational Physics in Granada, Spain.

\bibliography{bibliography}

\end{document}